\documentclass[a4paper]{IEEEtran}

\usepackage{soul}
\usepackage{cite}
\usepackage{amsmath,amssymb,amsfonts}
\usepackage{graphicx}
\usepackage{textcomp}
\usepackage{xcolor}
\usepackage[ruled,vlined]{algorithm2e}
\usepackage{algpseudocode}
\usepackage{bbm}
\usepackage{mathtools}
\usepackage{caption}
\usepackage{subcaption}
\usepackage{diagbox}
\usepackage{array}
\usepackage{amsmath}
\usepackage{amssymb}
\usepackage{amsfonts}
\usepackage{amsthm}
\usepackage{mathtools}

\usepackage{breqn}
\usepackage{enumitem}
\usepackage{physics}
\usepackage{thm-restate}
\usepackage{url}
\usepackage{tikz-cd}
\usepackage[normalem]{ulem}
\theoremstyle{definition}
\newtheorem{definition}{Definition}[section]
\theoremstyle{remark}

\DeclareMathOperator*{\argmin}{arg\,min}

\newcommand{\comment}[1]{}
\usepackage{authblk}
\usepackage[colorinlistoftodos]{todonotes}
\usepackage{stfloats}
\newcommand{\R}{\mathbb{R}}

\allowdisplaybreaks

\ifCLASSINFOpdf

\else

\fi

\usepackage{array}
\newcolumntype{L}[1]{>{\raggedright\let\newline\\\arraybackslash\hspace{0pt}}m{#1}}
\newcolumntype{C}[1]{>{\centering\let\newline\\\arraybackslash\hspace{0pt}}m{#1}}
\newcolumntype{R}[1]{>{\raggedleft\let\newline\\\arraybackslash\hspace{0pt}}m{#1}}

\begin{document}

\title{A privacy preserving querying mechanism with high utility for electric vehicles}

\author{Ugur Ilker Atmaca, 
        Sayan Biswas, \and
        Carsten Maple, \and 
        Catuscia Palamidessi 
        
\thanks{U. Atmaca and Carsten Maple are with the Secure Cyber Systems Research Group, WMG, University of Warwick \& Alan Turing Institute, UK, Emails: ugur-ilker.atmaca@warwick.ac.uk, cm@warwick.ac.uk}

\thanks{S, Biswas and C. Palamidessi are with the INRIA \& LIX, \'{E}cole Polytechnique, France, Emails: sayan.biswas@inria.fr, catuscia@lix.polytechnique.fr}
}


\markboth{
}%
{Atmaca \MakeLowercase{\textit{et al.}}: Bare Demo of IEEEtran.cls for IEEE Journals}

\IEEEpubid{
}
\maketitle

\begin{abstract}

Electric vehicles (EVs) are gaining popularity due to the growing awareness for a sustainable future. However, since there are disproportionately fewer charging stations than EVs, range anxiety plays a major role in the rise in the number of queries made along the journeys to find an available charging station. On the other hand, the use of personal data in various types of analytics is increasing at an unprecedented rate. Hence, the risks of privacy violation are also surging. Geo-indistinguishability is one of the standards for formalising location privacy as a generalisation of the local differential privacy. However, the noise has to be carefully calibrated considering the implications of potential utility-loss. In this paper, we introduce approximate geo-indistinguishability (AGeoI) which allows the EVs to obfuscate the individual query-locations while ensuring that they remain within their preferred area of interest. It is vital because journeys are often sensitive to a sharp drop in QoS, which has a high cost for the extra distance to be covered. We apply AGeoI and dummy data generation to protect the privacy of EVs during their journeys and preserve the QoS. Analytical insights and experiments are used to demonstrate that a very high percentage of EVs get privacy for free and that the utility-loss caused by the privacy-gain is minuscule. Using the iterative Bayesian update, our method allows for a private and highly accurate prediction of charging station occupancy without disclosing query locations and vehicle trajectories, which is vital in unprecedented traffic congestion scenarios and efficient route-planning.
\end{abstract}


\section{Introduction}




\IEEEPARstart{A}{ir} pollution is one of the immediate issues that the world is experiencing~\cite{FERRERO2016450,hampshire2018electric,Zhang_2020}. In the United Kingdom in 2019, 27\% of all greenhouse gas emissions come from transportation, as the largest emitting sector~\cite{hickman2007looking,kufeoglu2020emissions,Millen2021}. Hence, the transportation industry and academic communities are increasingly interested in developing alternative energy vehicles to reduce emissions. Automobile manufacturers are introducing a new generation of electric vehicles (EVs) that often employ connected and automated driving functions~\cite{gersdorf2020mckinsey}.

EVs are regarded as one of the most promising means of reducing emissions and reliance on fossil fuels. Along with environmental benefits, EVs provide superior energy efficiency to conventional vehicles~\cite{hensley2009electrifying}. As the cost of batteries continues to decrease, the large scale adoption of EVs is becoming more viable~\cite{gao2014road}. Despite the advantages and competitive cost, many customers remain concerned about running out of battery power before reaching their destination or waiting for their EVs to charge. The primary obstacles to EV adoption are the availability of chargers and the range that can be travelled on a single charge, often referred to as \emph{range anxiety} in the literature~\cite{lombardi2018electric}.

There has been some recent focus on forecasting how busy the charging stations (CS) are in certain areas to ensure that the EVs can plan their journeys conveniently~\cite{forecastingCS, DeepInformationFusion}. However, the existing research in this direction, primarily founded upon machine learning based methods, do not address the privacy concerns involved in such predictive techniques and do not consider situations where there may arise an unprecedented traffic congestion (e.g. due to a one-off concert or an event). One of the most successful approaches for protecting the privacy of personal data while analysing and exploiting the utility of data is along the lines of \emph{differential privacy} (DP)~\cite{DworkDP1,DworkDP2}, which mathematically guarantees that the query output does not change significantly regardless of whether a specific personal record is in the dataset or not. Our proposed method, in addition to allowing the EVs to have formal privacy guarantees on their queries to locate the nearest charging stations, enables the users to estimate the live occupancy of the charging stations efficiently allowing convenient journey-planning (e.g., avoid going through a route aiming for a CS which has a dense crowd of EVs looking for a CS around it).


However, the classical central DP requires a trusted curator who is responsible for adding noise to the data before publishing or performing analytics on it. A major drawback of such a central model is that it is vulnerable to security breaches because the entire original data is stored in a central server. Moreover, there is the risk of having an adversarial curator.  To circumvent the need for such a central dependency, a local model of differential privacy (DP), also called \emph{local differential privacy}~\cite{DuchiLDP}, has been getting a lot of attention lately. In this model, users apply the LDP mechanism directly to their data and send the locally changed data to the server.

LDP is particularly suitable for situations where users need to communicate their personal data in exchange for some service. One such scenario is the use of location-based services (LBS), where a user typically reports her location in exchange for information like the shortest path to a destination, points of interest in the surroundings, traffic information, friends nearby, etc. One of the recently popularised standards in location privacy is \emph{geo-indistinguishability} (GeoI)~\cite{AndresKostasCatuscia_GeoInd}, which optimises the quality of service (QoS) for users while preserving a generalised notion of LDP on their location data. The obfuscation mechanism of GeoI depends on the distance between the original location of a user and a potential noisy location that they report~\cite{Bordenabe:14:CCS,Fernandes:21:LICS}. 

GeoI can be implemented directly on the user's device (tablet, smartphone, etc.). The fact that the users can control their explicit privacy-protection level for various LBS makes it very appealing. However, a drawback of injecting noise locally to the datum is that it deteriorates the QoS due to the lack of accuracy of the data. 



On the other hand, future vehicles are getting more sophisticated in their sensory, onboard computation, and communication capacities. Furthermore, the emergence of Mobile Edge Computing (MEC) also changes the Intelligent Transportation Systems (ITS) by providing a platform to assist computationally heavy tasks by offloading the computation to the Edge cloud~\cite{gillam2018exploring}. This architecture often employs three tiers, with the vehicle on the first, MEC on the second, and standard cloud services on the third~\cite{maple2019connected}. Figure~\ref{fig:system_arch} shows the system architecture for the location privacy framework proposed in this paper. 

ITS provides a platform containing distributed and resource-constrained systems to support real-time vehicular functions where these functions' efficacy relies on the data shared across entities. However, the risk of privacy disclosure and tracking increases due to data sharing~\cite{hahn2019security}. Privacy-preserving schemes are developed using established techniques such as group signature, anonymity, and pseudonymity~\cite{lin2013achieving,kumar2015intelligent}. However, it is possible to identify privatised data with adequate background information. Hence, DP approaches have emerged as the gold standard of data privacy because they provide a formal privacy guarantee independent of a threat actor's background knowledge and computing capability~\cite{zhao2020survey}. 

GeoI is the state-of-the-art method for location privacy-preserving with LDP. It can preserve one's location privacy among a set of locations with similar probability distributions without requiring a trusted third-party. It provides rigorous privacy for location-based query processing and location data collection by modelling the location domain based on the Euclidean plane. However, vehicles are located on the road network under normal circumstances. For vehicular location queries, GeoI mechanism may result in publishing unrealistic privatised locations such as houses, parks, or lakes. Thus, there is a need for an adapted model of GeoI for vehicular application. This paper proposes a novel privacy model called Approximate Geo-Indistinguishability (AGeoI), based on the notion of GeoI by using a discrete road network graph. Our key contributions in this paper are outlined as follows.
\begin{enumerate}
    \item We presented the notion of AGeoI, a formal standard of location-privacy in a bounded co-domain, by generalising the classical paradigm of GeoI and adapting it to a graphical environment. We illustrate its applicability by proving that it satisfies the compositionality theorem. Moreover, we show that the truncated Laplace mechanism satisfies AGeoI by deriving the appropriate $\epsilon$ and $\delta$ as privacy parameters. 
    \item We proposed a two-way privacy-preserving navigation method for EVs dynamically querying for charging stations (CSs) on road networks — the method protects against threats to individual locations of their queries by ensuring AGeoI, and we provide protection against adversaries tracing the trajectories of the EVs by interpolating their query locations in an online setting.
    \item We performed experiments on real vehicular journey data from San Francisco with real locations of CSs from the area 
    under two settings of their sparsity in the road network, and show that our method ensures a very high fraction of EVs who enjoy \emph{privacy for free} and that the cost of utility-loss for the EVs is very low compared to the formal gain in privacy. 
    \item We demonstrated that our method, aside from providing formal location-privacy guarantees, allows the EVs to predict the live occupancy of the charging stations based on the sanitised queries received by the server, ensuring that the users can plan their journeys conveniently. 
\end{enumerate}

The rest of this paper is organised as follows. Section \ref{sec:related_work} reviews some of the related work in this area. Section \ref{sec:preliminaries} introduces some fundamental notions on DP and GeoI. Section \ref{sec:approx_geo_ind} develops the mathematical theory of AGeoI. Section \ref{sec:system_model} elucidates the model of our proposed mechanism by formalizing the problem we are tackling, thoroughly discussing system architecture, and laying out the privacy-threat landscape we are addressing in this work. Section \ref{sec:cost_of_privacy_analysis} delves into insight into the cost of privacy on the EVs induced by our mechanism. Section \ref{sec:experiments} presents the experimental results to illustrate the working of our mechanism, and Section~\ref{sec:conclusion} concludes the paper. 


\section{Related Work}\label{sec:related_work}


Both corporate and academic communities have recently piqued interest in advancing EVs and charging infrastructure to improve the transportation system's sustainability. Despite the advancements, the EV sector confronts challenges that delay the adoption process, such as range anxiety, an absence of convenient and available charging infrastructure, and waiting time to charge~\cite{franke2013understanding,kumar2020adoption}.
An offline static map of CSs is insufficient to resolve these obstacles since EVs may need to reserve a charging station when a trip is planned or query the available stations based on their battery state, and CSs must be reservable. Thus, live vehicular and charging station data is utilised in querying and reservation/scheduling mechanisms~\cite{tian2016real,zhang2021intelligent,flocea2022electric}. Encryption techniques can be used in such mechanisms to prevent external intrusions, but they cannot preserve users' privacy from malicious servers and third-party providers.   

Several data types are considered in these mechanisms, including real-time location, intended route, battery level, and station availability, to ensure the drivers are not detoured from their intended route\cite{tian2016real,wang2019sharedcharging}. Although disclosing such information poses privacy concerns for the driver's location and vehicle tracking, the privacy requirements of such mechanisms are not sufficiently studied in the literature. Existing methods for planning charging points for EV journeys are considered mechanisms for confidentiality and integrity, but the drivers' location privacy is regarded as an issue of trust in the third-party service provider~\cite{plugshare,chargepoint}. 

This problem can be addressed by several approaches based on threat model of the system. Location anonymity is achieved through cloaking an area~\cite{chow2009casper,niu2014achieving}. This approach can only be applied to the Edge of our system model to provide anonymity to a group of EVs, but we consider the Edge as an honest-but-curious threat actor and aim to preserve individual vehicles' privacy locally. Thus, such techniques are not trivially applicable to our considered threat model. Furthermore, anonymity techniques do not provide a formal privacy guarantee~\cite{DworkDP_Compositionality}. Similarly, mix-network approaches cannot be applied because there is no guarantee that multiple vehicles will be present in an Edge's coverage in any timestamp due to vehicles' movement~\cite{guo2018independent}. 

An applicable approach to download the charging station's live map on EVs to search for the nearest or on-the-route available charging station has been considered and studied by the community~\cite{CacheLocal}; however, the communication overhead of this technique is predicted to be much higher than the vehicles' location-based inquiry since it will require downloading a recent snapshot of the map for each query and, therefore, has been criticised in the literature~\cite{ChallengesOfLocal}. Moreover, due to the absence of data sharing, such methods hinder the statistical utility of the location data for the servers that may be useful for a variety of purposes (e.g. providing vital statistics to industries and institutions for optimally placing the charging stations on the map based on the query densities) and prevent the EVs from receiving any information about the traffic around and occupancy of certain charging stations restricting them to plan their journeys accordingly. 

DP methods are increasingly being deployed to preserve location privacy in a variety of domains, including automotive systems. The studies in~\cite{zhou2018achieving,luo2019geo} proposed models by deploying a GeoI-based mechanism on the Edge for location-based services. However, their approach did not consider preserving vehicles' location privacy against the Edge. An approach that compliments the problem we aim to address in this paper was proposed by Qiu et al. in \cite{qiu2020location} where the authors proposed a technique to crowd-source a task in vehicular network while preserving GeoI of the location of the vehicles offering Mobility as a Service in the spatial network to solve a task at a publicly known location in the map (e.g. taxi services). The problem formulation in this work is the inverse of what we aim to achieve in this paper. As a result, the work in \cite{qiu2020location} cannot be extended to address the privacy concerns induced my multiple queries dynamically made along the journey. 

In \cite{Cunningham2021Trace}, Cunningham et al. studied the problem of trajectory sharing under DP and proposed a mechanism to tackle it. However, this work assumes the setting of an offline trajectory sharing which breaks down in the practical environment where the trajectories are being shared online as there is no prior information or limitation on the number of queries made by an EV during a journey, and their respective locations. Therefore, the method proposed by the authors in \cite{Cunningham2021Trace} cannot be directly adapted to our dynamic environment closely simulating the real-world scenario for such a use case. 

Of late, a major direction of research is along the lines of studying the statistical utility of differentially private data. A standard notion  of statistical utility, which is extended to a variety of contexts, is the precision of the estimation of  the distribution on the original data from that of the noisy data. Iterative Bayesian update (IBU)~\cite{AgarwalIBU,agrawal2005privacy}, an instance of the 
famous \emph{expectation maximization (EM)} method from statistics, provides one of the most flexible and powerful estimation techniques  and has recently become in the focus of the community~\cite{EhabConvergenceIBU, EhabGIBU}. In this work, we use IBU to approximate the distribution of the true locations of the queries made to the server and based on that, the users of the EVs can predict the availability of the the CSs around them in real time and plan their route accordingly.




\section{Preliminaries}\label{sec:preliminaries}

The most successful approach to formally address the privacy risks is DP, mathematically guaranteeing that the query output does not change significantly regardless of whether a specific personal record is in a dataset or not. Most research performed in this area probes two main directions. One is the classical central framework~\cite{DworkDP1,DworkDP2}, in which a trusted third-party (the curator) collects the users' personal data and obfuscates them with a differentially private mechanism. 
\begin{definition}[Differential privacy~\cite{DworkDP1,DworkDP2}]
\label{def:dp}
For a certain query, a randomizing mechanism $\mathcal{R}$ provides \emph{$\epsilon$-DP} if, for all neighbouring\footnote{differing in exactly one place} datasets, $D$ and $D'$, and all $S \subseteq$ Range($\mathcal{R}$), we have
$
\mathbb{P}[\mathcal{R}(D) \in S] \leq e^{\epsilon}\,\mathbb{P}[\mathcal{R}(D') \in S]
$.
\end{definition}

A major drawback of central model is that it is vulnerable to security breaches because the entire original data is stored in a central server. Moreover, there is the risk that the curator may be corrupted. Therefore, a local variant of the central model has been widely popularized of late~\cite{DuchiLDP}, where the users apply a randomizing mechanism locally on their data and send the perturbed data to the collector such that a particular value of a user's data does not have a major probabilistic impact on the outcome of the query. 

\begin{definition}[Local differential privacy~\cite{DuchiLDP}]
\label{def:ldp}
Let $\mathcal{X}$ and $\mathcal{Y}$ denote the spaces of original and noisy data, respectively. A randomizing mechanism $\mathcal{R}$ provides \emph{$\epsilon$-LDP} if, for all $x,\,x'\,\in\,\mathcal{X}$, and all $y\,\in\,\mathcal{Y}$, we have
$
\mathbb{P}[\mathcal{R}(x)=y] \leq e^{\epsilon}\,\mathbb{P}\left[\mathcal{R}(x')=y\right]
$.
\end{definition}

Recently, GeoI~\cite{AndresKostasCatuscia_GeoInd}, a variant of the local DP capturing the essence of the distance between locations~\cite{Bordenabe:14:CCS,Fernandes:21:LICS} has been in focus as a standard for privacy protection for location based services, being motivated by the idea of preserving the best possible quality of service despite the local obfuscation operated on the data.

\begin{definition}[Geo-indistinguishability~\cite{AndresKostasCatuscia_GeoInd}]
\label{def:geoind}
Let $\mathcal{X}$ be a space of locations and let $d_{\text{E}}(x,x')$ denote the \emph{Euclidean distance} between $x\in \mathcal{X}$ and $x'\in \mathcal{X}$. A randomizing mechanism $\mathcal{R}$ is \emph{$\epsilon$-geo-indistinguishable} if for all $x_1,\,x_2\in\mathcal{X}$, and every $y\in \mathcal{X}$, we have
$
\mathbb{P}[\mathcal{R}(x)=y] \leq e^{\epsilon d_{\text{E}}(x_1,x_2)}\,\mathbb{P}\left[\mathcal{R}(x')=y\right]
$.
\end{definition}


\begin{definition}[Iterative Bayesian update~\cite{AgarwalIBU}]\label{def:IBU}
Let $\mathcal{C}$ be a privacy mechanism that locally obfuscates points from the discrete space $\mathcal{X}$ to $\mathcal{Y}$ such that $\mathcal{C}_{xy}=\mathbb{P}(y|x)$ for all $x,y\in \mathcal{X},\mathcal{Y}$. Let $X_1,\ldots,X_n$ be i.i.d. random variables on $\mathcal{X}$ following some distribution $\pi_{\mathcal{X}}$. Let $Y_i$ denote the random variable of the output when $X_i$ is obfuscated with $\mathcal{C}$. 

Let $\vb{y}\in\mathcal{Y}^n$ be a realisation of $\{Y_1,\ldots,Y_n\}$ and $q$ be the empirical distribution obtained by counting the frequencies of each element of $\mathcal{Y}$ as observed in $\vb{y}$. The \emph{iterative Bayesian update (IBU)} is a state-of-the-art EM technique to estimate $\pi_{\mathcal{X}}$ by converging to the maximum likelihood estimate (MLE) of $\pi_{\mathcal{X}}$ with the knowledge of $\vb{y}$ and  $\mathcal{C}$. IBU works as follows:
\begin{enumerate}
\item Start with any full-support PMF $\theta_0$ on $\mathcal{X}$.

\item Iterate $ \theta_{t+1}(x)=\sum\limits_{y\in\mathcal{Y}}q(y)\frac{\theta_t(x)\mathcal{C}_{xy}}{\sum\limits_{z\in\mathcal{X}}\theta_t(z)\mathcal{C}_{zy}}$ for all $x\in\mathcal{X}$.
\end{enumerate}
\end{definition}

\begin{table}
\caption{List of key notations}\label{table:notations}
\centering
\resizebox{\columnwidth}{!}{%
\begin{tabular}{cc}
\hline
Notation & Description \\ \hline
 
$\mathcal{X}$  & Domain of original locations \\
$d_{\mathcal{X}}$ & Distance on $\mathcal{X}$\\
$\mathcal{Y}$  & Domain of obfuscated locations \\
$d_{\mathcal{Y}}$ & Distance on $\mathcal{Y}$\\
$\mathbb{P}_{\mathcal{K}}\left[y\vert x\right]$ & Prob. that mechanism $\mathcal{K}$, applied to value $x$, reports $y$\\
$I$ & Fixed Edge in the network\\
$R(I)$ & Area of coverage by $I$\\
$m$ & Number of locations reported by each EV\\
$l_u$ & Vector of locations reported by EV $u$\\
$\mathcal{L}(t)$ & Set of location vectors received by $I$ at time $t$\\
$\mathcal{L}'(t)$ & Shuffled set of all individual locations queried at time $t$\\
$\mathcal{R}(t)$ & Set of nearest CSs for $\mathcal{L}'(t)$\\
$G$ & Road network graph\\
$d_G$ & Travelling distance in graph $G$\\
\hline
\end{tabular}%
}
\end{table}




\section{Approximate geo-indistinguishability}\label{sec:approx_geo_ind}

In the classical framework of GeoI~\cite{AndresKostasCatuscia_GeoInd}, the space of the noisy data is, in theory, unbounded under the planar Laplace mechanism. Under a certain level of GeoI that is achieved, the planar Laplace mechanism ensures a non-zero probability of obfuscating an original location to a privatised one which may be quite far, thus inducing a possibility of a substantial deterioration in the QoS of the users. This loss of QoS can be more sensitive in the context of the navigation of EVs, where it is extremely important to prioritize a bounded domain where a user is willing to drive -- this may be a result of time constraint, rising cost of fuel, geographical boundaries (e.g. international borders), etc. -- giving rise to an idea of \emph{area of interest} for each EV. This motivated us to extend the classical GeoI to a more generalized, approximate paradigm, inspired by the approach of the development of approximate DP from its pure counterpart.

Let $\mathcal{X}$ and $\mathcal{Y}$ be the spaces of the real and nosy locations equipped with distance metrics $d_{\mathcal{X}}$ and $d_{\mathcal{Y}}$, respectively. In general $\left(\mathcal{X},d_{\mathcal{X}}\right)$ and $\left(\mathcal{Y},d_{\mathcal{Y}}\right)$ may be different and unrelated. However, for simplicity, here we assume $\mathcal{X}\subseteq \mathcal{Y}$ and, therefore, $d_{\mathcal{X}}=d_{\mathcal{Y}}=d$, and we proceed to define the notion of \emph{approximate geo-indistinguishability}. Note that $d$ may not necessarily need to be symmetric, i.e., there may exist $x1,\,x_2\,\in\,\mathcal{Y}$ such that $d(x_1,x_2)\neq d(x_2,x_1)$.

\begin{definition}[Approximate geo-indistinguishability] A mechanism $\mathcal{K}$ is \emph{approximately geo-indistinguishable (AGeoI)} or $\emph(\epsilon,\delta)$-\emph{geo-indistinguishable} if for any set of values $S\subseteq\mathcal{Y}$ and any pair of values $x,x'\,\in\,\mathcal{X}$, there exists some $\epsilon,\,\delta\in\mathbb{R}_{\geq 0}$ such that $\delta e^{d(x,x')}\in[0,1]$, we have:
\begin{equation}\label{eq:approxGeoI}
    \mathbb{P}_{\mathcal{K}}\left[y\,\in\, S\vert x\right]\leq e^{\epsilon\,d(x,x')}\mathbb{P}_{\mathcal{K}}\left[y\,\in\, S\vert x'\right]+\delta\,e^{d(x,x')}
\end{equation}
\end{definition}

One of the biggest advantages of DP and all of its variants that are accepted by the community is the property of compositionality, where the level of privacy can be formally derived with repeated number of queries. Thus, we now enable ourselves to investigate the working of the compositionality theorem with the approximate geo-indistinguishability which we defined, to stay consistent with the literature~\cite{DworkDP_Compositionality}. 

\begin{restatable}{theorem}{compositionality}\label{th:compositionality}[Compositionality Theorem for AGeoI]
Let mechanisms $\mathcal{K}_1$ and $\mathcal{K}_2$ be $(\epsilon_1,\,\delta_1)$ and $(\epsilon_2,\,\delta_2)$ geo-indistinguishable, respectively. Then their composition is $(\epsilon_1+\epsilon_2,\,\delta_1+\delta_2)$-geo-indistinguishable. In other words, for every $S_1,\,S_2\subseteq \mathcal{Y}$ and all $x_1,\,x'_1,\,x_2,\,x'_2\,\in\,\mathcal{X}$, we have:
\begin{align}\label{eq:composition}
    \mathbb{P}_{\mathcal{K}_1,\mathcal{K}_2}\left[(y_1,y_2)\,\in\, S_1\times S_2\vert (x_1,x_2)\right]\nonumber\\
    \leq e^{\epsilon_1\,d(x_1,x'_1)+\epsilon_2\,d(x_2,x'_2)}\nonumber\\
    \times\mathbb{P}_{\mathcal{K}_1,\mathcal{K}_2}\left[(y_1,y_2)\,\in\, S_1\times S_2\vert (x'_1,x'_2)\right]\nonumber\\
    +\left(\delta_1+\delta_2\right)\,e^{d(x_1,x'_1)+d(x_2,x'_2)}
\end{align}
\end{restatable}

\begin{proof}
Let us simplify the notation and denote: $$P_i=\mathbb{P}_{\mathcal{K}_i}\left[y_i\in S_i\vert x_i\right]$$ $$P'_i=\mathbb{P}_{\mathcal{K}_i}\left[y_i\in S_i\vert x'_i\right]$$ 
$$\tilde{\delta}_i=\delta_i\,e^{d(x_i,x'_i)}$$
for $i\,\in\,\{1,2\}$.
As mechanisms $\mathcal{K}_1$ and $\mathcal{K}_2$ are applied independently, we have:
\begin{align}
    \mathbb{P}_{\mathcal{K}_1,\mathcal{K}_2}\left[(y_1,y_2)\,\in\, S_1\times S_2\vert (x_1,x_2)\right]=P_1.P_2 \label{eq:simp1}\\
    \mathbb{P}_{\mathcal{K}_1,\mathcal{K}_2}\left[(y_1,y_2)\,\in\, S_1\times S_2\vert (x'_1,x'_2)\right]=P'_1.P'_2 \label{eq:simp2}
\end{align}
Therefore, we obtain:
\begin{align}
     \mathbb{P}_{\mathcal{K}_1,\mathcal{K}_2}\left[(y_1,y_2)\,\in\, S_1\times S_2\vert (x_1,x_2)\right]=P_1.P_2\nonumber\\
     \leq\left(\min\left(1-\delta_1,e^{\epsilon_1\,d(x_1,x'_1)}P'_1\right)+\tilde{\delta}_1\right)\nonumber\\
     \times\left(\min\left(1-\delta_2,e^{\epsilon_2\,d(x_2,x'_2)}P'_2\right)+\tilde{\delta}_2\right)\nonumber\\
     \leq m_1m_2+\tilde{\delta}_1m_2+m_1\tilde{\delta}_2+\tilde{\delta}_1+\tilde{\delta}_2\nonumber\\
     \left[\text{where }m_i=\min\left(1-\delta_i,e^{\epsilon_i\,d(x_i,x'_i)}P'_i\right)\right]\nonumber\\
     \leq e^{\epsilon_1\,d(x_1,x'_1)+\epsilon_2\,d(x_2,x'_2)}P'_1P'_2\nonumber\\
     +\tilde{\delta_1}-\tilde{\delta_1}\tilde{\delta_2}+\tilde{\delta_2}-\tilde{\delta_1}\tilde{\delta_2}+\tilde{\delta_1}\tilde{\delta_2}\nonumber\\
     \leq e^{\epsilon_1\,d(x_1,x'_1)+\epsilon_2\,d(x_2,x'_2)}\nonumber\\
     \times\mathbb{P}_{\mathcal{K}_1,\mathcal{K}_2}\left[(y_1,y_2)\,\in\,
     S_1\times S_2\vert (x'_1,x'_2)\right]\nonumber\\
     +\left(\delta_1+\delta_2\right)\,e^{d(x_1,x'_1)+d(x_2,x'_2)}\nonumber
\end{align}
\end{proof}

We now proceed to generalize the conventional planar Laplace mechanism~\cite{ChatzikokolakisElSalamounyPalamidessi+2017+308+328} to define the \emph{truncated Laplace mechanism} extended to a generic metric space. 

\begin{definition}[Truncated Laplace mechanism]\label{def:TruncatedLaplace}
The \emph{truncated Laplace mechanism} $\mathcal{L}$ on a space $\mathcal{X}$ equipped with, not necessarily symmetric, distance metric $d$ truncated to a radius $r$, is defined as:
\begin{align}\label{eq:TruncatedLaplace}
 \mathbb{P}_{\mathcal{L}}[y\vert x]
    =\begin{cases}
		    c\,e^{-\epsilon\,d(y,x)} &  \text{if }d(x,y)\leq r\\
            0 & \text{otherwise}
		\end{cases}
\end{align}
where $c$ is the truncated normalization constant defined such that $\int\limits_{y\in\mathcal{Y}}\mathbb{P}_{\mathcal{L}}[y\vert x]dy=1$, and $\epsilon$ is the desired privacy parameter.  Let us call $r$ to be the \emph{radius of truncation} for $\mathcal{L}$.
\end{definition}

Note: In case of a discrete domain $\mathcal{Y}$, $c$ is defined by normalizing $\sum\limits_{y\in\mathcal{Y}}\mathbb{P}_{\mathcal{L}}[y|x]=1$, and in this case $\mathcal{L}$ is an extended truncated geometric mechanism \cite{ghosh2009universally} extended to a generic metric space.

\begin{restatable}{lemma}{positivedelta}\label{lem:positivedelta}
    $\frac{e^{-\epsilon\,d(x_1,x_2)}\mathbb{P}\left[y\vert x_1\right]-\mathbb{P}\left[y\vert x_2\right]}{e^{(1-\epsilon)\,d(x_1,x_2)}}\,e^{r}\leq 1$, where $r$ is the radius of truncation for $\mathcal{L}$, as in \eqref{eq:approxGeoI}.
\end{restatable}

\begin{proof}
    \begin{align}
        \frac{e^{-\epsilon\,d(x_1,x_2)}\mathbb{P}\left[y\vert x_1\right]-\mathbb{P}\left[y\vert x_2\right]}{e^{(1-\epsilon)\,d(x_1,x_2)}}\,e^{r}\leq 1\nonumber\\
        \iff c\left(e^{-\epsilon d(x_1,x_2)+d(x_1,y)}-e^{-\epsilon d(x_2,y)}\right)\nonumber\\ 
        \leq e^{(1-\epsilon)d(x_1,x_2)-r}\label{eq:legaldelta}
    \end{align}
Now we observe that $d(x_1,x_2)+d(x_1,y)\geq d(x_2,y)$ due to the fact that $d$ is a metric and it satisfies triangle inequality. Immediately, we have $e^{-\epsilon d(x_1,x_2)+d(x_1,y)}-e^{-\epsilon d(x_2,y)}\leq0$ for any $\epsilon\in\mathbb{R}_{\geq 0}$. Therefore, as $c\geq 0$, \eqref{eq:legaldelta} is trivially satisfied as the RHS is always non-negative.
\end{proof}

\begin{restatable}{theorem}{truncatedLap}\label{th:truncatedLap}
$\mathcal{L}$ satisfies $(\epsilon,\delta)$-geo-indistinguishability where
$\delta=\max\left\{\max\limits_{\substack{S\subseteq \mathcal{Y}\\x_1,x_2\in\mathcal{X}}}\frac{e^{-\epsilon\,d(x_1,x_2)}\mathbb{P}\left[y\vert x_1\right]-\mathbb{P}\left[y\vert x_2\right]}{e^{(1-\epsilon)\,d(x_1,x_2)}},0\right\}$.
\end{restatable}

\begin{proof}
Trivially $\delta e^{d(x_1,x_2)}>0$ for any $x_1,x_2\in\mathcal{X}$ as $\delta>0$. Moreover, Lemma \ref{lem:positivedelta} ensures that  $\delta e^{d(x_1,x_2)}<1$.

Now observe that for every $S\subseteq\mathcal{Y}$ and for all $x_1,x_2\,\in\,\mathcal{X}$, we have:
\begin{align}
    e^{-\epsilon\,d(x_1,x_2)}\mathbb{P}_{\mathcal{L}}\left[y\vert x_1\right]-\mathbb{P}_{\mathcal{L}}\left[y\vert x_2\right]\leq \delta\,e^{(1-\epsilon)\,d(x_1,x_2)}\nonumber\\
    \implies \mathbb{P}_{\mathcal{L}}\left[y\vert x_1\right]-e^{\epsilon\,d(x_1,x_2)}\mathbb{P}_{\mathcal{L}}\left[y\vert x_2\right]
    \leq \delta\,e^{d(x_1,x_2)}\nonumber
\end{align}
\end{proof}

The explicit process of sampling private locations satisfying AGeoI from a given set of original locations through truncated Laplace mechanism on a discrete location space has been described in Algorithms \ref{alg:AGeoIMech} and \ref{alg:AGeoISamp}. 
\begin{algorithm}
\SetAlFnt{\small}
\textbf{Input:} Discrete domain of original locations: $\mathcal{X}$,  Discrete domain of private locations: $\mathcal{Y}$, Desired privacy parameter: $\epsilon$, Desired truncation radius: $r$;
\textbf{Output:} Channel $C$ satisfying \eqref{eq:TruncatedLaplace}\;
\SetKwFunction{Fmain}{DTLap}
\SetKwProg{Fn}{Function}{:}{}
\Fn{\Fmain{$\mathcal{X},\mathcal{Y},\epsilon,r$}}{ 
Set $C\leftarrow \text{empty channel}$\;
Set $Y\leftarrow \text{empty list}$\;
\For{$x \in \mathcal{X}$}
{
    $c_x=\frac{1}{\sum\limits_{\substack{y\in\mathcal{Y}\\d(x,y)\leq r}}e^{-\epsilon\,d(x,y)}}$ \;
    \For{$y \in \mathcal{Y}$}
    {  
        \eIf{$d(x,y)\leq r$}
        {$C[x,y]=0$}
        {$C[x,y]=c_x\,e^{-\epsilon\,d(x,y)}$} 
    }
}
\textbf{Return:} $C$\;
}
\caption{Discrete and truncated Laplace mechanism (DTLap)}\label{alg:AGeoIMech}
\end{algorithm}

\begin{algorithm}
\SetAlFnt{\small}
\textbf{Input:} Discrete domain of original locations: $\mathcal{X}$,  Discrete domain of private locations: $\mathcal{Y}$, Desired privacy parameter: $\epsilon$, Desired truncation radius: $r$; Vector of original locations: $X$\;
\textbf{Output:} Corresponding vector of private locations: $Y$\;
\SetKwFunction{Fmain}{DTLapSamp}
\SetKwProg{Fn}{Function}{:}{}
\Fn{\Fmain{$\mathcal{X},\mathcal{Y},\epsilon,r,X$}}{ 
$C=$ \Call{DTLap}{$\mathcal{X},\mathcal{Y},\epsilon,r$}\;
Set $Y\leftarrow \text{empty list}$\;
\For{$x \in X$}
{
    Randomly sample $y\in\mathcal{Y}\sim C[x,:]$\;
    Append $y$ to $Y$
}
\textbf{Return:} $Y$\;
}
\caption{Sampling private locations with DTLap (DTLapSamp)}\label{alg:AGeoISamp}
\end{algorithm}

\section{System Model}\label{sec:system_model}

This section details our privacy-preserving model for finding an optimal charging station in the Internet of Vehicles (IoV) as a use case of the proposed AGeoI technique. We begin with a discussion of the location privacy problems inherent in finding optimal CSs in the IoV. This is followed by road networking modelling, a description of the system architecture for differentially private location sharing, the trust relationship between system tiers, and the privacy threat model.

\subsection{Problem Statement}

EVs have been identified as a critical component of future sustainable transportation systems to reduce CO2 emissions and have garnered significant attention from academia and business~\cite{kumar2020adoption}. Due to the limited battery capacity, EVs may be required to visit CSs during journeys. It causes some drivers to have range anxiety, which is the vehicle has insufficient battery power to cover the travel needed to reach its intended destination. It is highlighted as one of the barriers to EVs' widespread adoption~\cite{bulut2017mitigating}. CSs can not always be readily available since it often takes a while to be sufficiently charged for EVs. A CS booking service can help to overcome range anxiety.

EVs may access such services through third-party providers to discover the nearest and readily available CSs to minimise charging wait times by using static or live location queries. However, location sharing raises privacy challenges, such as vehicle tracking. GeoI technique provides a formal privacy guarantee for location queries. However, it is not highly applicable to this use case for two reasons. It does not consider the feasible locations where a vehicle can be present, and it does not stop vehicle tracking in the case of linked queries during the vehicle trajectory. Thus, a tailored privacy-preserving mechanism is facilitated by combining the proposed AGeoI technique with dummy location generation. 

\subsection{Road Network Model}
Similar to~\cite{qiu2020location}, the road network $G$ is represented as a weighted directed graph $G=(N,E,W)$, where $N$ is the set of nodes, $E\subseteq N^2$ is the set of edges, and $W:N^2\rightarrow \mathbb{R}^{+}$ is the set of weights representing the minimum travelling distance between any two nodes. The nodes and edges correspond to junctions and road segments of the network, respectively. Each edge $e \in E$ is addressed by the pair of respective starting node, ending node, and a weight representing the travelling distance through that edge, i.e., $e=(N^s_e, N^e_e, w_e) \in N$, where the direction of the traffic is from $N^s_e$ to $N^e_e$ on $e$. For any $i\in N$ and $j\in N$, let the sequence of edges $(e_1,\ldots,e_r)$ denote a \emph{path} from node $i$ to node $j$ if $N^s_{e_1}=i$ and $N^e_{e_r}=j$. Hence, let $C(i,j)$ represent the set of paths that connect node $i$ to node $j$. Then $W$ is a $N\times N$ matrix, where 
\begin{equation*}
   W_{ij}=
   \begin{cases}
        \min\limits_{p\in C(i,j)}\sum\limits_{e\in p}w_e & \text{ if $C(i,j)\neq \phi$}\\
        \infty & \text{ otherwise}
   \end{cases}
\end{equation*}
Essentially $W_{ij}$ is the shortest travelling distance from node $i$ to node $j$ in the network. We shall address the quantity $W_{ij}$ as the \emph{traversal distance} between nodes $i$ and $j$ in the graph $G$ and denote it as $d_{G}(i,j)$ for every $(i,j)\,\in N^2$. Note that, as $G$ is a directed graph, $d_{G}$ may not be symmetric. 

\subsection{System Architecture}
\begin{figure}[ht!]
  \centering
  \includegraphics[width=0.48\textwidth]{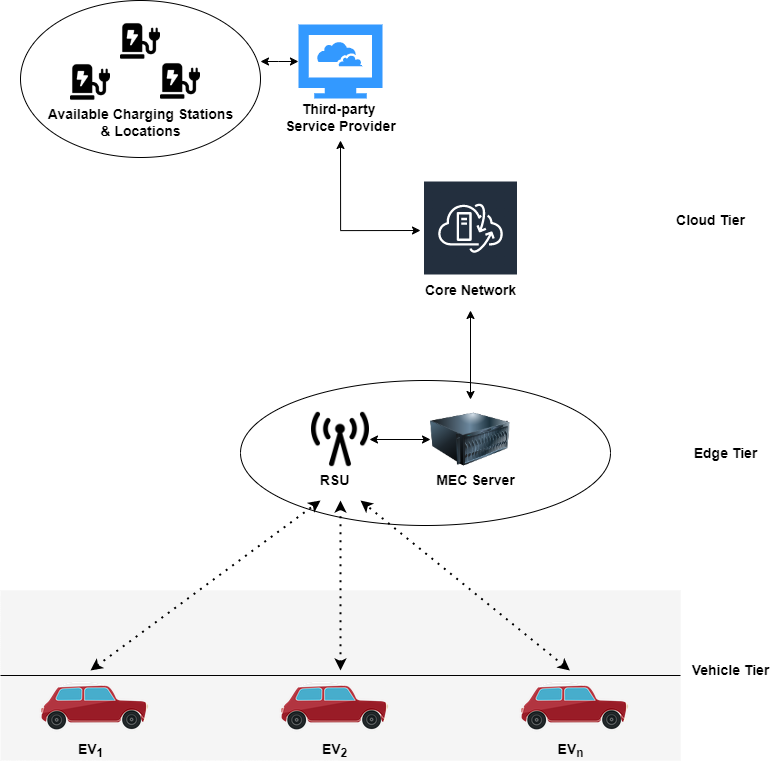}
  \caption{System Architecture (EV:Electric Vehicle, RSU: Roadside Unit, MEC: Mobile-Edge Computing Unit)}
  \label{fig:system_arch}
\end{figure}

IoV applications are transforming transportation systems by reducing human error, increasing travel convenience, and lowering energy, operational, and environmental costs~\cite{Duan2020emerging,ji2020Artificial}. Electric vehicles (EVs) have emerged as a viable technology for lowering carbon emissions and travel costs~\cite{patil2020grid}. However, range anxiety is one of the major challenges of their wide adoption. Vehicular location data can be utilised to optimise the vehicle charging plan and mitigate range anxiety. Third-party services such as Chargemap~\cite{chargemap} and PlugShare~\cite{plugshare} can recommend available and the closest CSs for the users. However, users are required to trust these third-party service providers to use these services, which presents significant privacy concerns in the honest-but-curious service provider threat model.


As the system architecture is depicted in Figure~\ref{fig:system_arch}, the vehicles are part of a three-tier cloud computing architecture and connected to Roadside Units (RSUs), MEC Server, and the Core Cloud Network through a secure communication channel, where these components are considered trusted. The Core Cloud Network enables the connection between the vehicle and third-party services, including charging station recommender, as the main consideration of this paper. However, it is not possible to ensure fully trusted third-party services that they will not utilise vehicular location data for further objectives. Thus, the honest-but curious threat model is considered for the third-party service provider and only privatised vehicular location data is shared in our proposed architecture. The roles of the system components are described in the following.

\subsubsection{Vehicle Tier}
In this paper, we fix a road network $G$ with nodes $G(N)$ and edges $G(E)$. We choose an arbitrary edge $I\in G$, and focus on the queries made by the EVs in $I$'s range of coverage, $R(I)$, provided by its RSU. When an EV moves from the area of coverage of one Edge cloud to another, we can assume the queries and the privacy threats against the Edge to reset as each Edge communicates with the Cloud-based services and the third-party service providers. 



\subsubsection{Edge Tier} Given the large volume of data generated and shared between vehicles and infrastructure, installing Edge cloud close to the vehicles are needed to host the off-board vehicular services, which require a low access latency from the onboard vehicular services~\cite{Lee2018exploring}. Along with essential data processing and forwarding functions, the Edge tier provides another layer for data aggregation and deploying additional privacy-preserving measures before sharing the data with any third-parties. 

\subsubsection{Cloud Tier} It is expected to provide computation and storage capabilities for top-level processes, including data sharing interfaces for third-party services. 

\subsubsection{Third-party Service Provider} It is the external party to ITS and is expected to enhance the quality of the function for finding the available CSs for the vehicles by receiving search queries compromised of privatised and dummy location vectors for the respective vehicles.

\subsubsection{Communication Channel} ITS comprises a network of RSU, vehicle on-board electronic control units (ECU), and distributed cloud computing and storage services. Wireless communications are enabled for V2V (Vehicle-to-Vehicle), V2I (Vehicle-to-Infrastructure) and V2X(Vehicle-to-Everything) by the technologies such as IEEE 802.11p DSRC/WAVE (Dedicated Short Range Communication/Wireless Access in Vehicular Environments), cellular advances such as C-V2X, and the long-term evolution for vehicles (LTE-V)~\cite{seo2016lte}. Confidentiality of the wireless communication channel is secured by public key infrastructure (PKI) encryption methods which is beyond the scope of this work. 

\subsection{Privacy Threat Landscape}

In real-time IoV location-based applications, vehicle users who wish to receive services tailored to their current locations need to share their location with the service provider. Different applications may have varying needs for data sensitivity. Data perturbation techniques are often used in private data sharing by introducing uncertainty in the data, which entails a trade-off in terms of the related application's quality of service in general. For example, EV users who want to locate the nearest available charging station during travel can send a query with their current location to a service provider using the Intelligent Transportation Systems (ITS) infrastructure. 


We consider the system into three categories: (i) the vehicle users (data subject), (ii) ITS infrastructure, including the MEC and Cloud Tiers (data controller and data processor), and (iii) the third-party that receives the privatised data from the deployed deployment privacy-preserving mechanisms. The third-party is assumed to be an EV charging management system, which may operate in a registration-based approach for a specific area.


The vehicular location data may also be subject to unauthorised use, data inference, retention or disclosure besides its primary use. Thus, a privacy-preserving mechanism can be utilised for location data sharing to provide a privacy guarantee while the users can receive the same or similar quality of service. Thus, the third-party provider is considered an honest-but-curious adversary model, which assumes it is honest in accurately executing the protocol required to provide location data. However, it may be curious to infer users' private information based on the obtained location data~\cite{paverd2014modelling}.


In this work, we mainly address two notable sources of threat that might put the EVs' privacy at risk:

\textbf{Location identification:} It is essential that the locations of the queries made by the EVs are protected from being identified. For the use case of location CSs, it is essential that the privatized version of the true location of the EV is within a certain radius of interest w.p. 1, making sure that the reported location is within a feasible and drivable distance away, and most importantly, within the area of coverage of the Edge where its true location lies. Therefore, we defined AGeoI as an extension of GeoI, which is the standard for location privacy. Thus, to ensure the privacy of the position of any given query in the road network, the EVs locally obfuscate their true locations using the truncated Laplace mechanism with their desired parameter $\epsilon$ and the radius of truncation $r$, which, in turn, decide the value of $\delta$.

\textbf{Journey tracing:} 
It is often the case that a certain EV might enquire about the nearest available charging station, but choose to proceed with the journey at that point of time, and raise further queries along the journey. In our model, we capture this realistic setting by allowing multiple queries to be made by the EV in a given journey. This immediately leads to a threat of approximately tracing the trajectory of the journey of an EV by interpolating the locations of the query, despite having each individual location being AGeoI. 

This is due to the fact that the obfuscated location of each query is not distinguishable from the real location, but they are not too far off from each other with a very high probability. Therefore, if the number of queries made in a single journey is large, it could be fairly straightforward to approximately deduce the trajectory of the journey. 

Cunningham et al.~\cite{Cunningham2021Trace} proposed a mechanism to securely share trajectories under LDP. However, the authors in \cite{Cunningham2021Trace} assumed a model of offline sharing of the entire trajectory and, hence, sanitising it with the proposed mechanism to engender a LDP guarantee. In our setting, this method cannot be directly implemented as we are working in a dynamic environment where the queries made by the EVs are in real time, with the server not having any prior knowledge of the number or the location of the queries made by a certain EV. Therefore, the mechanism of \cite{Cunningham2021Trace} cannot trivially be extended in the online location-sharing environment, and hence, the threat of adversaries able to reconstruct the journey of a particular EV with a high number of queries remains as a concern. 

\begin{figure*}[ht!]
  \centering
  \includegraphics[width=0.9\textwidth]{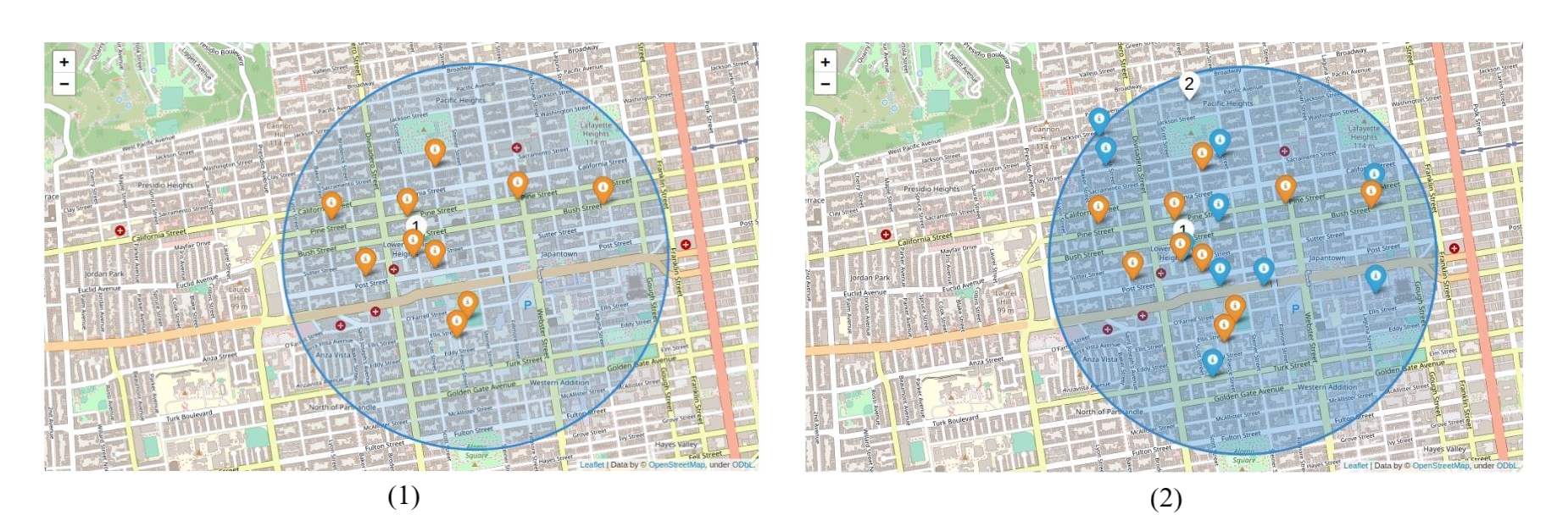}
  \caption{Reported dummy and privatised locations for two respective time windows (White Pins: Privatised locations, Orange Pins: Dummy locations in $1^{\text{st}}$ Time window, Blue Pins: Dummy locations in $2^{\text{nd}}$ Time window)}
  \label{fig:dummylocations}
\end{figure*}

\begin{figure}[ht!]
  \centering
  \includegraphics[width=0.37\textwidth]{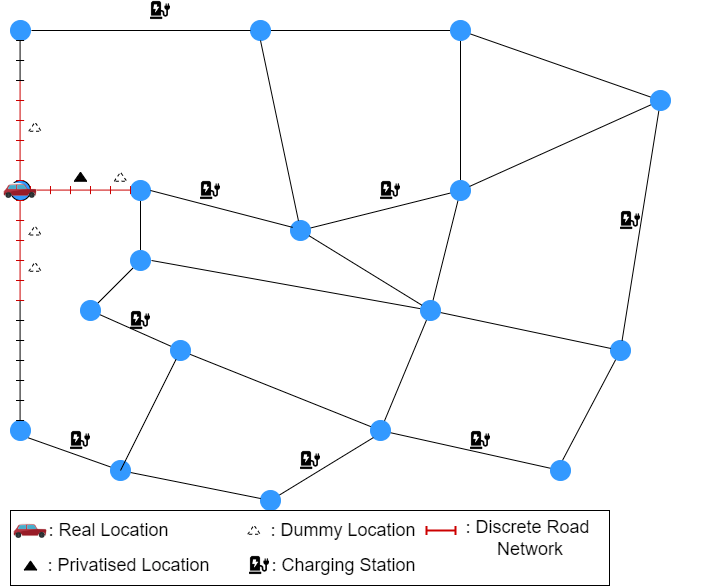}
  \caption{A toy example for a static location query on discrete road network}
  \label{fig:toy1}
\end{figure}

\begin{figure}[ht!]
  \centering
  \includegraphics[width=0.37\textwidth]{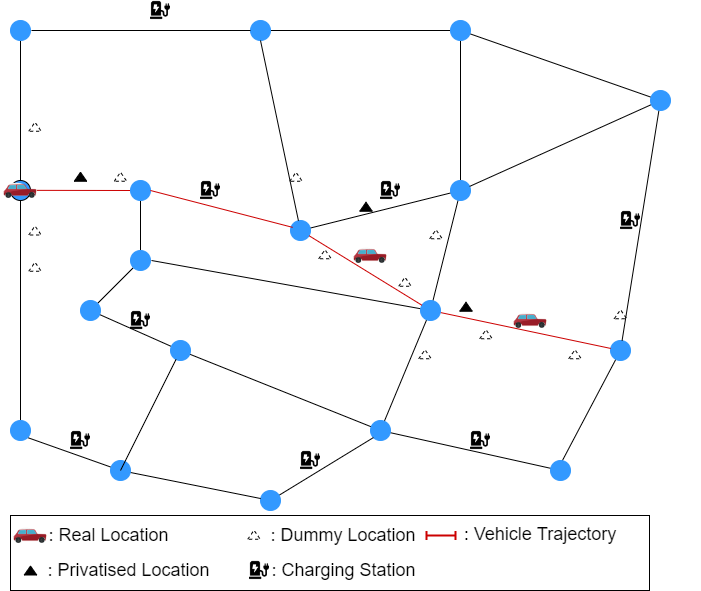}
  \caption{A toy example for linked 3 location queries on discrete road network}
  \label{fig:toy2}
\end{figure}

\subsection{Proposed Method}

At every query made along the journey, an EV $u$ located within the coverage of an Edge $I$ locally obfuscates its true location $x^u\in R(I)$ to $x^u_1\in R(I)$ using the truncated Laplace mechanism, guaranteeing AGeoI, and generate $m-1$ feasible dummy locations $\{x^u_2,\ldots,x^u_m\}\in R(I)^{m-1}$, i.e., locations that cannot be trivially identified as being artificially generated given the query of the previous time stamp w.r.t. realistic speed limits, travelling conditions, etc. For the first query that $u$ makes along its journey, it can generate any random $m-1$ dummy locations in $R(I)$. Thus, each CS query of $u$ consists of reporting the vector of $m$ locations, $l_u=\left(x^u_1,\ldots,x^u_m\right)\in\R(I)^m$, to $I$ for the Edge to process and communicate the query to the Cloud services and the third-parties to find the nearest available CSs in $R(I)$ for $u$. This inductively ensures that the adversary will have at least $m$ possible trajectories that the EV could have realistically followed at every time stamp, making it highly improbable for the Edge and the third-party to be able to conclude which of them was the actual journey as, after $k$ queries made along a single journey, each interpolated trajectory will have a probability of at least $1/m^k$ of being the real one.

Figure~\ref{fig:dummylocations} shows an example of reported $10$ dummy locations with the privatised location for two respective time windows, where the dummy locations in the following time window can have a feasible link at least one of the preceding dummy locations. Toy examples are depicted on Figures~\ref{fig:toy1} and \ref{fig:toy2} as more simplistic explanations of the proposed privacy-preserving mechanism, where the first represents a static query on an exampler discrete road network, and the following one represents linked dynamic queries on the same exampler discrete road network.
 
At any time, the Edge collects all the reported locations from the querying EVs, shuffles them by effacing the links between the location vectors and the corresponding EVs, and sends this jumbled collection of all the reported locations in the network for knowing their respective nearest available CSs to the third-party service provider. After receiving the response, the Edge, which internally keeps the record of the IDs of the EVs against their queried locations, assigns the corresponding vector of locations of the nearest available CSs to each EV and communicates them back to the respective vehicles. 

In other words, at time $t$, if the Edge receives the location vectors from $k_t$ querying EVs as $\mathcal{L}(t)=\{l_{u_1},\ldots, l_{u_{k_t}}\}$, the Edge is responsible for shuffling all the individual location points in these reported vectors and forward the scrambled collection of locations $\mathcal{L}'(t)=\{x^u_i\colon u\in\{u_1,\ldots,u_{k_t}, i\in[m]\}$ to the Cloud and the third-party, while internally keeping a track of the IDs of the EVs to reconnect the query-response back to the corresponding users. Setting $\hat{x}$ as the location of the nearest available charging station from location $x$ in $R(I)$, the Edge receives $\mathcal{R}(t)=\{\hat{x}^u_i\colon u\in\{u_1,\ldots,u_{k_t}, i\in[m]\}$ as the response from the third-party service provider handling the CS data real-time. After this, matching the stored IDs of the EVs with the locations of the CSs, the Edge communicates the response vector $\hat{l}_u=(\hat{x}^u_i:i\in [m])$ back to the corresponding EV $u$. Then the EV can choose to navigate to $\argmin_{x\in\hat{l}_u}\{d_G(x,x_u)\}$, where $x_u$ is the real location of EV $u$. The overview of this mechanism is explained in Figure~\ref{fig:system_arch}.

\section{Cost of privacy analysis} \label{sec:cost_of_privacy_analysis}

\begin{definition}[Cost of privacy]\label{def:CoP}
Suppose an EV $u$ at location $x^u$ chooses to locally obfuscate its real location of query as $x^u_1$ using the truncated Laplace mechanism $\mathcal{L}_{\epsilon,r}$ satisfying $(\epsilon,\delta)$-geo-indistinguishability with a corresponding radius of truncation $r$. Then we define the \emph{cost of privacy (CoP)} of EV $u$ as $\operatorname{CoP}(u,\mathcal{L}_{\epsilon,r})=\vb{c}(x^u,\hat{x}^u_1)-\vb{c}(x^u,\hat{x}^u)$, where $\hat{x}^u$ and $\hat{x}^u_1$ are the nearest available CSs in the network to $x^u$ and $x^u_1$, respectively, and $\vb{c}:G(N)^2\mapsto \mathbb{R}^+$ is any cost function that reflects the cost of commute from locations $x$ to $y$ in the network. 
\end{definition}

In other words, CoP, as in Definition \ref{def:CoP}, essentially captures the \emph{extra} cost that an EV needs to cover as a result of the privatized location it reports to the Edge satisfying AGeoI, as opposed to its true location. For the purpose of this paper and simplicity of the analysis, we considered the cost function as the travelling distance in the network, i.e., $\vb{c}=d_G$. However, in practice, any suitable cost function could be used (e.g. fuel efficiency, time, etc.) could be used as $\vb{c}$, depending on the context and requirement of the architecture.

To formally characterize and analyze the CoP of the EVs in the network, inspired from the classical version of \emph{Voronoi decomposition}, we extend the concept in the setting of our road network in the network coverage for a fixed Edge w.r.t. graph-traversal distance, $d_g$.

\begin{definition}[Voronoi decomposition]\label{def:Voronoi}
Let $G$ be the graph representing the road network equipped with travelling distance $d_G$. Let the set of CSs in $G$ be $C_G=\{c_1,\ldots,c_{n_G}\}$. Then the \emph{Voronoi decomposition} on $G$ w.r.t. $C_G$ is defined as $\vb{V}_G=\{V_i\colon\,i\in[n_G]\}$ such that $V_i\cap V_j=\phi$ for any $i\neq j$ and $\bigcup\limits_{i\in[n_G]}V_i=G$, where
\small
\begin{equation*}
  V_i=\{x\in G\colon\,d_G(x,c_i)\leq d_G(x,c_j)\,\forall\,j\in[n_G],j\neq i\}  
\end{equation*}
\par
\end{definition}

\begin{definition}[Closed ball around a location]\label{def:ClosedBall}
For any $x \in G$ and $r\in\mathbb{R}_{\geq 0}$, the \emph{closed ball} of $x$ of radius $r$ is defined as $\beta_r(x)=\{y\in G\colon\,d_G(x,y)\leq r\}$
\end{definition}

\begin{definition}[Fenced Voronoi decomposition]\label{def:FenceVoronoi}
For any $r\in\mathbb{R}_{\geq 0}$ and charging station $i$, let the \emph{$r$-fenced Voronoi decomposition} on road network $G$ be defined as $V^{-r}_G=\{V_i^{r}\colon\,i\in[n_G]\}$ such that  $V^{-r}_i\cap V^{-r}_j=\phi$ for $i\neq j$ and $V_i^{-r}=\{x\in V_i \colon\, B_r(x)\subseteq V_i\}$. In other words, $V_i^{-r}$ essentially constructs an area contained within $V_i$ restricted by a \emph{fence} at a distance $r$ from the edge of $V_i$. 
\end{definition}

\begin{restatable}{theorem}{zeroCoP}\label{th:zeroCoP} Suppose an EV $u$ positioned at $x^u$ on $G$ obfusucates its location using AGeoI with any radius of truncation $r\in\mathbb{R}_{\geq 0}$. Let $\hat{x}^u$ be the location of the nearest available charging station to the true location $x^u$. Then $\mathbb{P}\left[\operatorname{CoP}(u,\mathcal{L}_{\epsilon,r})=0\right]=1$ for every $x^u\in V_{\hat{x}^u}^{-r}$. In other words, if an EV lies in the $r$-fenced Voronoi decomposition for its nearest available CS, it has a \emph{zero cost for privacy} w.p. 1. 
\end{restatable}
\begin{proof}
Immediate from Definition~\ref{def:FenceVoronoi}.
\end{proof}

\begin{restatable}{theorem}{zeroCoPProb}\label{th:zeroCoPProb} Suppose an EV $u$ lies in $V_{\hat{x}^u}\setminus V_{\hat{x}^u}^{-r}$ and it uses AGeoI to obfuscate its true location $x^u$ to $x^u_1$ with a radius of truncation $r$ and privacy parameter $\epsilon$ for making a private query to the Edge. Then $\mathbb{P}\left[\operatorname{CoP}(u,\mathcal{L}_{\epsilon,r})=0\right]=1-\sum_{x^u_1\in V_{\hat{x}^u}^c}ce^{-\epsilon d_G(x^u,x^u_1)}$, where $c$ is the normalizing constant of the truncated Laplace mechanism as in Definition \ref{def:TruncatedLaplace}.
\end{restatable}

\begin{proof}
To compute $\mathbb{P}\left[\operatorname{CoP}(u,\mathcal{L}_{\epsilon,r})=0\right]$, we only need to exclude the possibilities where the reported location of the EV lies outside the Voronoi decomposition of the station $\hat{x}^u$, which, essentially, is $1-\sum_{x^u_1\in V_{\hat{x}^u}^c}ce^{-\epsilon d_G(x^u,x^u_1)}$.
\end{proof}

\section{Experimental Study}\label{sec:experiments}


A series of experiments are conducted to evaluate the proposed method for the use case of an EV querying for a CS in Edge cloud-assisted ITS. The experiments' goals are to (i) validate proposed theoretical claims and solutions empirically; (ii) use the method to find the nearest available charging station for EVs as a case study; (iii) investigate the cost of privacy in real-time settings; and (iv) conduct a real-time CS occupancy prediction technique from the noisy vehicle distribution. Standard Python packages are used to run the experiments in an environment with an Intel core i7 processor, 16 GB of RAM, and an Ubuntu 20.04 operating system. 

\subsection{Dataset Preparation}

The road network, exported from OpenStreetMap~\cite{OpenStreetMap}, is utilised in this study. The cost of privacy is analysed by calculating the extra routing distance to the identified optimal charging station due to additional noise in vehicular locations in the queries. Thus, the cost of privacy depends on the sparsity of the CSs. Two datasets with different densities are prepared to comprehend the impact of this sparsity. The United States Department of Energy allows for the download of up-to-date data on existing and planned alternative fuel stations, as well as their geographical coordinates~\cite{chargingstations}. The first dataset is exported, containing the locations of 404 existing CSs in San Francisco. For the second dataset, we followed the assumption that CSs will be more widely spread in the future, and parking locations are likely to deploy charging facilities. Thus, the second dataset is prepared by merging the existing and planned charging station locations with on-street and off-street parking locations provided by DataSF~\cite{chargingstations2}. The second dataset contains 716 independently distributed locations.

The EPFL (Swiss Federal Institute of Technology-Lausanne) mobility dataset contains the GPS records of the 536 taxi trajectories in San Francisco for four weeks. The data format includes a taxi identifier, latitude, longitude, state of the taxi (vacant or occupied), and a UNIX epoch timestamp~\cite{piorkowski2009crawdad}. It is possible to split the entire taxi trajectory into realistic trajectories of each customer because the dataset contains information on the occupancy of the taxis. By doing so, we could export over 450 thousand trajectories. Among those, 536 trajectories are randomly chosen from each taxi for this study. 

\subsection{Experimental Setup}

In the experimental scenario, a group of vehicles sends out location queries to find the closest available CSs during their journeys on the road network $G$. The edges of the road network $G$ are truncated into discrete segments with an equal $k$ travel distance, similar to the work in~\cite{qiu2020location}. DTLap (e.g., Algorithm 1) is utilised to generate the privacy channel by using the Laplace mechanism for the user's desired values for privacy budget $\epsilon$ and truncation radius $r$. Following this, DTLapSamp (e.g., Algorithm 2) is used to generate privatised locations with respect to the users' real locations. 

A location query contains a privatised location and $m-1$ dummy locations as a vector and is collected by the Edge for sending them to the third-party through the core cloud as a single vector of all locations. The third-party responds to the locations in the vector with the closest available CS for each, and the Edge sends vehicle location vectors to the related vehicles without being able to differentiate privatised and dummy locations.

For IBU to approximate the original distribution of the query locations of the EVs in the road network in order to predict the availability of the CSs and, thus, assist the users in planning their journeys appropriately, we note that each original query location goes through two independent steps of sanitization: a) locally using the truncated Laplace mechanism to achieve AGeoI and b) generating the realistic dummy locations in the area of coverage of the Edge to ensure protection against attacks reconstructing their journeys. Setting the domain $\mathcal{X}$ as the area of coverage of the RSU of the fixed Edge that we focus on, while the former is a straightforward use of the channel $\mathcal{L}$, the latter can be thought of as $m-1$ independent applications of the uniform channel $\mathcal{U}$, where $U:\mathcal{X}^2\mapsto \mathbb{R}$ with $\mathcal{U}_{x,y}$ denoting $\mathbb{P}_{\mathcal{U}}[y|x]=1/|\mathcal{X}|$, by each EV. Therefore, after accounting for the normalization, the channel incorporating the local obfuscation and the generation of the dummy locations used by each EV reduces down to $\frac{1}{m}\mathcal{L}+\frac{m-1}{m}\mathcal{U}$ which we use as the privacy channel to implement IBU.

The first set of experiments examines the CoS for randomly selected 536 vehicle traces, where each trace contains a series of GPS coordinates and 3 randomly selected points along each for the real locations of the queries. The discrete road network is generated by setting the distance $k=100$ meters. The parameters of $\epsilon$ and $r$ are varied in the range of 0.2 to 2, and 1 to 20, respectively.

The privatised location, together with the dummy locations, is sent to the third-party for a query to prevent the third-party from tracking the vehicle. The area of Edge coverage, rather than the vehicle's area of interest, is considered for dummy location generation rather than the vehicle's area of interest, as the centre of mass may give away the true location. The second set of experiments examined the impact of dummy locations on the CoS. 

The location queries could be used for real-time predictive analysis on the optimisation of the smart power grid, managing staff, and determining where new CSs should be deployed. Thus, service providers can have the utility of the datasets (e.g., train ML models, etc.) with DP-based methods while the privacy of individuals is preserved. The third set of experiments utilises the IBU method to retrieve the true distribution of locations queries from the noisy distribution, which includes privatised and dummy locations.



\begin{figure*}[t]
  \centering
  \includegraphics[width=0.78\textwidth]{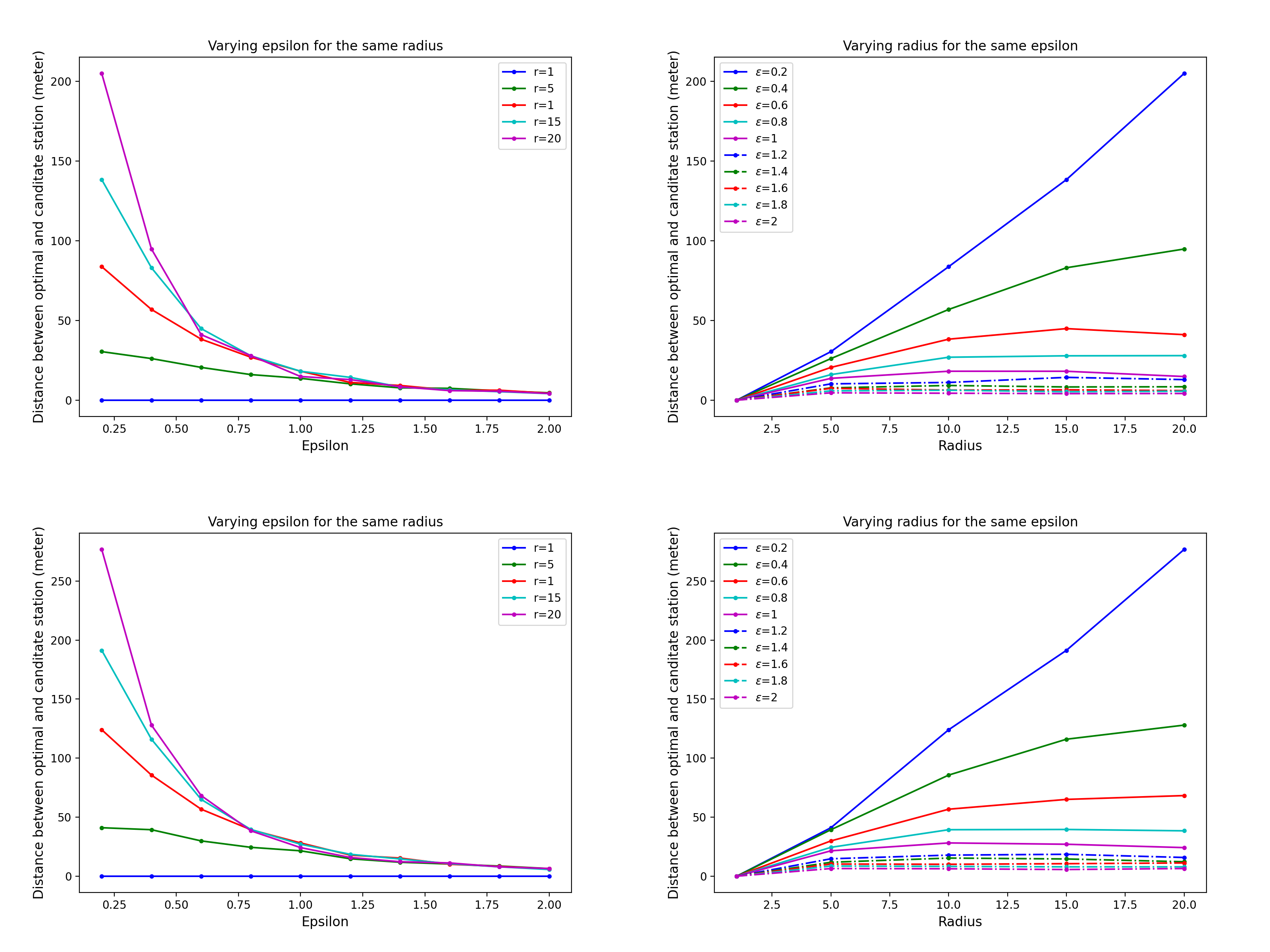}
  \caption{CoS for varying $\epsilon$ or $r$ of AGeoI ($1st$ row is for sparse CSs, $2nd$ row is for dense CSs)} 
  \label{fig:cost_of_privacy}
\end{figure*}

\begin{figure*}[ht!]
  \centering
  \includegraphics[width=1\textwidth]{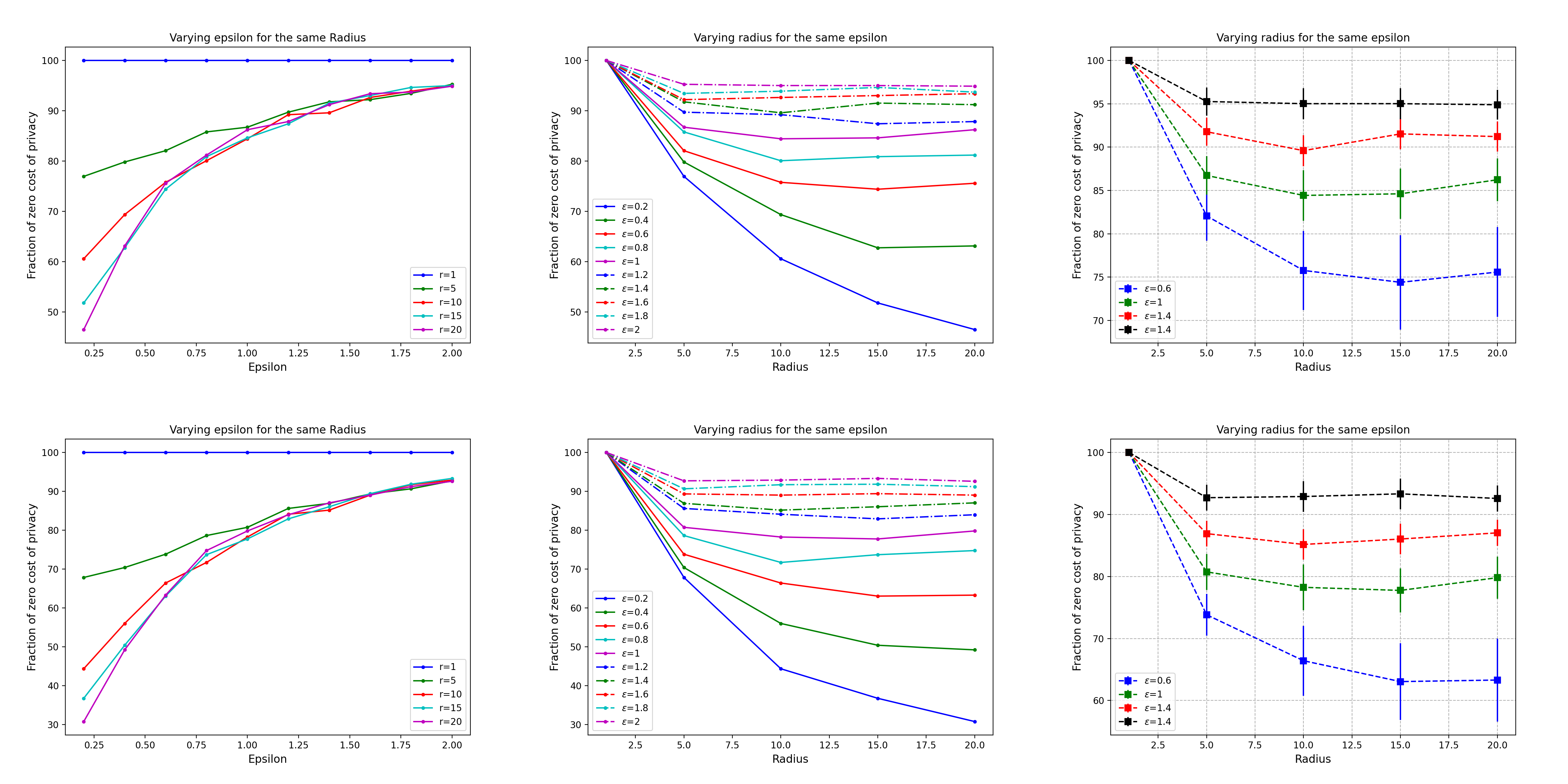}
  \caption{Fraction of zero CoS for varying $\epsilon$ or $r$ of AGeoI ($1st$ row is for sparse CSs, $2nd$ row is for dense CSs)}
  \label{fig:privacy_for_free}
\end{figure*}

\begin{figure*}[ht!]
  \centering
  \includegraphics[width=0.8\textwidth]{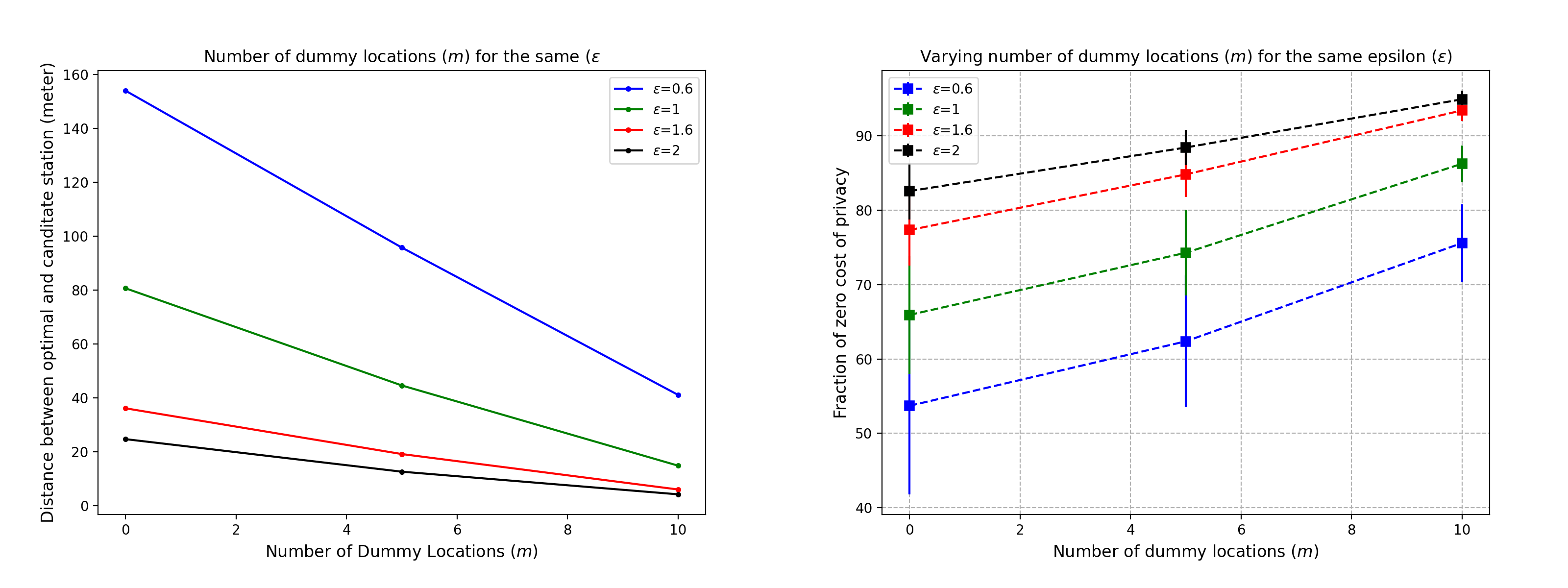}
  \caption{Impact of introducing dummy locations along
with AGeoI on the CoS}
  \label{fig:impact_of_dummydata}
\end{figure*}

\begin{figure}[ht!]
  \centering
  \includegraphics[width=0.4\textwidth]{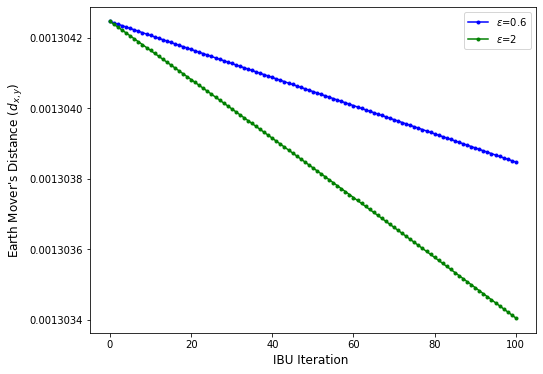}
  \caption{The distance between the original and estimated distributions using IBU for $\epsilon=0.6$ and $\epsilon=2$ noisy distributions by the earth mover’s distance (i.e. Kantorovich-Wasserstein distance)}
  \label{fig:ibu}
\end{figure}

\begin{figure}[ht!]
  \centering
  \includegraphics[width=0.48\textwidth]{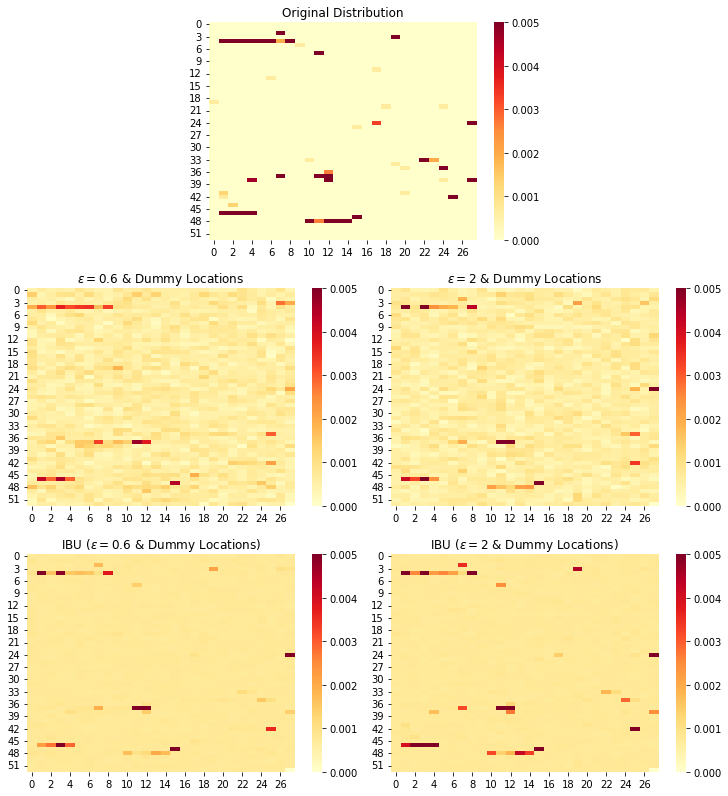}
  \caption{Estimations of the original distribution using IBU for the $\epsilon=0.6$ and $\epsilon=2$ noisy distributions.}
  \label{fig:ibu_heat}
\end{figure}

\subsection{Results and Discussion}

\subsubsection{Cost of Privacy Analysis}
DP approaches introduce a trade-off between privacy and data utility, with a higher level of privacy requiring a greater level of noise. The efficacy of the respective service may correspondingly decrease due to the fall in data utility, and it is considered the `cost of privacy' (CoS) in this study. However, there is a finite number of cases with maximum efficacy for many applications. Thus, there must be a sweet spot of privacy for free where the user will receive the maximum efficacy with some level of privacy. It is also valid for our use case, in which a vehicle must be able to utilise a finite number of unique CSs. 

The following results are achieved by carrying out the experiments for $3$ linked queries of $536$ randomly selected vehicle trajectories for varying values of $\epsilon$ or $r$ ranging from $0.2$ to $2$, and $1$ to $20$, respectively. Figure~\ref{fig:cost_of_privacy} demonstrates the CoS in terms of the extra travelling distance due to the privacy-preserving mechanism, where a similar pattern is observed for both of the datasets. Another observation is that a high frequency of queries resulted in no cost for privacy preservation. Figure~\ref{fig:privacy_for_free} shows the fraction of the queries with ``privacy for free'' where both the datasets followed similar patterns. Vehicle queries contain dummy locations and their privatised true locations. It is possible that the dummy locations can sometimes provide a better utility, but our experiments consider the utility of a privatised location as the worst-case for analysis.

Figure~\ref{fig:cost_of_privacy} shows that our method provides a negligible cost of utility-loss for the formal privacy-gain enjoyed by the EVs. By increasing the truncation radius, an abrupt drop in the distance between the location of the nearest available charging station for the true location of the query and that of the privatised one, implies that the cost of the extra travel distance needed to be taken due to the AGeoI guarantee is almost negligible. A similar trend is seen for the varying $\epsilon$ with a fixed radius. As the level of privacy decreases, the fraction of EVs in the network enjoying \emph{privacy for free} grows to be more than $60\%$ for a radius of truncation of merely 10 road segments, where each segment is 100 meters long, for $\epsilon\geq 0.5$. However, more than $90\%$ of the EVs achieve a zero cost of privacy for $\epsilon\geq 1.5$, irrespective of the truncation radius as illustrated in Figure~\ref{fig:privacy_for_free}. Due to increasing perturbation in disclosed location, the width of the confidence interval for zero cost of privacy increases, as seen in Figure~\ref{fig:privacy_for_free}. The likelihood of achieving zero cost of privacy fluctuates over a wider range and it does not monotonically decrease with the growing radius due to rising randomness.

\subsubsection{Impact of Dummy Data Generation}
Considering an adversary interested in finding the true locations of the EVs, $(\alpha, \beta]$\emph{-identifiability} is defined for any location $x$ as $\mathbb{P}[d(y,x) < \alpha)>\beta$, where $y$ is any guessed location by the adversary. With proposed method, with a sufficiently small radius of truncation to obfuscate the true location using the truncated Laplace and generating $m-1$ dummy locations in the area of coverage of the Edge, the probability of hitting the true $x$ within an error of $\alpha$ is $\mathbb{P}[d(x,y)<\alpha]= \frac{1}{m}c e^{-\epsilon \alpha} =\beta$, where $c$ is the normalising constant.

There has been some work in this area from the perspective of just GeoI~\cite{qiu2020location, luo2019geo, zhou2018achieving, Li_PerturbationHidden, zhang2018privacy} or just from the standpoint of generating dummy locations exploiting anonymisation techniques~\cite{Freudiger:109437,Jiang_LPPM}. One of the first major concerns in using only GeoI is when we allow dynamic and multiple queries along the journey of the EVs, as individual locations, despite being privatised, can still be interpolated to approximate the entire trace. If only dummy locations are used, however, any estimated (or observed) $y$ could be the real location w.p. $\frac{1}{m-1}$, as there is no formal privacy guaranteed, i.e., every location $x$ has, is $(0,1/m-1)$-identifiable among $(m-1)$ dummy locations. 
With potential parallel processing, brute-force attacks are just one way that it has been shown that anonymisation techniques are not sufficient to protect privacy~\cite{zang2011anonymization}.

Figure~\ref{fig:impact_of_dummydata} illustrates how the CoP increases with an increase in the noise due to lack of dummy locations under the same level of identifiability. To achieve the same $(\alpha,\beta)$-identifiability with just AGeoI without dummy locations, the parameter $\epsilon$ needs to be scaled by $\frac{1}{\ln{m}}$, i.e., more noise needs to be added, which results in having a worse trade-off between privacy and CoP for the same level of privacy.  

\subsubsection{Real-time Predictive Study}
Predicting the availability of CSs is a crucial component of EV trip planning and can help ease range anxiety. SoA methods utilised machine learning-based approach to develop such prediction models~\cite{hecht2021predicting,nait2018prediction, ma2022multistep, almaghrebi2020data,luo2021deep}.The main consideration of these models is that drivers can book timeslots for CSs and the prediction is made based on factors such as past usage of CCs, traffic density, and some other features such as weather conditions. However, such consideration may be limited to facilitating effective EV journey planning, given that traffic is highly dynamic and subject to unexpected changes. Due to the traffic, EVs may be late for their scheduled charging time, and another EV cannot be navigated to charge from the same station, despite the fact that it is empty. Hence, a real-time predictive analysis would be critical to determine the likelihood of a CS being available when an EV arrives. Our proposed method provides privacy-preserved live traffic distribution of the querying vehicles. IBU method~\cite{EhabGIBU,EhabConvergenceIBU} is utilised to demonstrate the statistical distance between the predicted and the original traffic distribution for our results with $\epsilon=0.6$ and $\epsilon=2$, see Figure~\ref{fig:ibu} for IBU with 100 iterations where the distance is decreasing by iterations of IBU. Although such a statistical distance, it is more important to predict the heavy hitters, which helps the prediction of how likely a CS will be available when the vehicle arrives, see Figure~\ref{fig:ibu_heat} for the heatmaps of original, noisy, and predicted traffic distributions. 




\section{Conclusion}\label{sec:conclusion}

This paper studied a fundamental problem of the risk of privacy violation for EVs dynamically querying for charging stations along their journeys. The setting of the problem has not been addressed in the literature, and some of the related techniques along the lines of privacy-preserving vehicle routing cannot be adapted directly into the practical model considered in this paper.

To resolve this, we theorised the notion of AGeoI and justified its mathematical soundness and applicability by proving the compositionality theorem, which is a novelty in itself and allows us to attain GeoI in a strictly bounded space of noisy data. We derived the appropriate privacy parameters to prove that the truncated Laplace mechanism satisfies AGeoI and used it to propose a location privacy-preserving method for EVs querying for CSs. Our method protects privacy in two ways: at the specific positions of the queries and along the entire journey.

In experimental studies, datasets with real vehicle traces and locations were used to demonstrate the trade-off between privacy and utility and the impact of dummy locations on this trade-off. We used IBU for real-time prediction of the original distribution of the EVs from the reported (noisy) locations. The proposed method is distinct from current machine learning-based approaches in that it considers real-time changes in the number of location-based queries. Thus, it can reflect unprecedented traffic variations at the CSs occupancy rate. By using IBU, the method is capable of predicting the likelihood of a particular station being occupied by another vehicle at the time of arrival and, hence, enables an online prediction technique to estimate the availability of CSs around an EV, allowing the users to do convenient route-planning. A consistent trend of a substantial majority of the EVs to have privacy for free was observed across all the experiments, i.e., most of the EVs suffer no loss of utility even for fairly high-level formal AGeoI. In general, we observe that the cost of privacy induced by our method is fairly low across settings, thus, ensuring formal privacy protection for the location of the EVs without incurring a high price for that as far as the journey is concerned. We dissected this cost of privacy under our method using Voronoi decomposition to draw insight into the workings of our method.



\section*{Acknowledgment}
The authors would like to thank Professor Graham Cormode for the insightful discussions that supported the development of this study. Ugur Ilker Atmaca and Sayan Biswas are shared co-first authors who have worked together on this paper and contributed equally. The authors would also like to extend their sincere thanks to the European Research Council (ERC) Advanced Grant of Catuscia Palamidessi which supported Sayan Biswas’ research visit at WMG, the University of Warwick, which enabled this successful collaboration.

\bibliographystyle{IEEEtran}
\bibliography{references.bib}

\begin{thebibliography}{10}
\providecommand{\url}[1]{#1}
\csname url@samestyle\endcsname
\providecommand{\newblock}{\relax}
\providecommand{\bibinfo}[2]{#2}
\providecommand{\BIBentrySTDinterwordspacing}{\spaceskip=0pt\relax}
\providecommand{\BIBentryALTinterwordstretchfactor}{4}
\providecommand{\BIBentryALTinterwordspacing}{\spaceskip=\fontdimen2\font plus
\BIBentryALTinterwordstretchfactor\fontdimen3\font minus
  \fontdimen4\font\relax}
\providecommand{\BIBforeignlanguage}[2]{{%
\expandafter\ifx\csname l@#1\endcsname\relax
\typeout{** WARNING: IEEEtran.bst: No hyphenation pattern has been}%
\typeout{** loaded for the language `#1'. Using the pattern for}%
\typeout{** the default language instead.}%
\else
\language=\csname l@#1\endcsname
\fi
#2}}
\providecommand{\BIBdecl}{\relax}
\BIBdecl

\bibitem{FERRERO2016450}
E.~Ferrero, S.~Alessandrini, and A.~Balanzino, ``Impact of the electric
  vehicles on the air pollution from a highway,'' \emph{Applied Energy}, vol.
  169, pp. 450--459, 2016.

\bibitem{hampshire2018electric}
K.~Hampshire, R.~German, A.~Pridmore, and J.~Fons, ``Electric vehicles from
  life cycle and circular economy perspectives,'' \emph{Version}, vol.~2,
  p.~25, 2018.

\bibitem{Zhang_2020}
R.~Zhang and S.~Fujimori, ``The role of transport electrification in global
  climate change mitigation scenarios,'' \emph{Environmental Research Letters},
  vol.~15, no.~3, p. 034019, Feb 2020.

\bibitem{hickman2007looking}
R.~Hickman and D.~Banister, ``Looking over the horizon: Transport and reduced
  co2 emissions in the uk by 2030,'' \emph{Transport Policy}, vol.~14, no.~5,
  pp. 377--387, 2007.

\bibitem{kufeoglu2020emissions}
S.~Kufeoglu and D.~K.~K. Hong, ``Emissions performance of electric vehicles: A
  case study from the united kingdom,'' \emph{Applied Energy}, vol. 260, p.
  114241, 2020.

\bibitem{Millen2021}
\BIBentryALTinterwordspacing
E.~P. Seymour~Millen, ``Transport and environment statistics 2021 annual
  report,'' May 2021. [Online]. Available:
  \url{https://assets.publishing.service.gov.uk/government/uploads/system/uploads/attachment_data/file/984685/transport-and-environment-statistics-2021.pdf}
\BIBentrySTDinterwordspacing

\bibitem{gersdorf2020mckinsey}
T.~Gersdorf, P.~Hertzke, P.~Schaufuss, and S.~Schenk, ``Mckinsey electric
  vehicle index: Europe cushions a global plunge in ev sales,'' 2020.

\bibitem{hensley2009electrifying}
R.~Hensley, S.~Knupfer, and D.~Pinner, ``Electrifying cars: How three
  industries will evolve,'' \emph{McKinsey Quarterly}, vol.~3, no. 2009, pp.
  87--96, 2009.

\bibitem{gao2014road}
P.~Gao, R.~Hensley, and A.~Zielke, ``A road map to the future for the auto
  industry,'' \emph{McKinsey Quarterly, Oct}, pp. 1--11, 2014.

\bibitem{lombardi2018electric}
M.~Lombardi, K.~Panerali, S.~Rousselet, and J.~Scalise, ``Electric vehicles for
  smarter cities: the future of energy and mobility,'' in \emph{World Economic
  Forum. http://www3. weforum. org/docs/WEF\_2018\_\%
  20Electric\_For\_Smarter\_Cities. pdf}, 2018.

\bibitem{forecastingCS}
A.~Ostermann, Y.~Fabel, K.~Ouan, and H.~Koo, ``Forecasting charging point
  occupancy using supervised learning algorithms,'' \emph{Energies}, vol.~15,
  no.~9, p. 3409, May 2022.

\bibitem{DeepInformationFusion}
A.~Sao, N.~Tempelmeier, and E.~Demidova, ``Deep information fusion for electric
  vehicle charging station occupancy forecasting,'' in \emph{2021 IEEE
  International Intelligent Transportation Systems Conference (ITSC)}, 2021,
  pp. 3328--3333.

\bibitem{DworkDP1}
C.~Dwork, F.~McSherry, K.~Nissim, and A.~Smith, ``Calibrating noise to
  sensitivity in private data analysis,'' in \emph{Theory of Cryptography},
  S.~Halevi and T.~Rabin, Eds.\hskip 1em plus 0.5em minus 0.4em\relax Berlin,
  Heidelberg: Springer Berlin Heidelberg, 2006, pp. 265--284.

\bibitem{DworkDP2}
C.~Dwork, K.~Kenthapadi, F.~McSherry, I.~Mironov, and M.~Naor, ``Our data,
  ourselves: Privacy via distributed noise generation,'' in \emph{Advances in
  Cryptology - EUROCRYPT 2006}, S.~Vaudenay, Ed.\hskip 1em plus 0.5em minus
  0.4em\relax Berlin, Heidelberg: Springer Berlin Heidelberg, 2006, pp.
  486--503.

\bibitem{DuchiLDP}
J.~C. Duchi, M.~I. Jordan, and M.~J. Wainwright, ``Local privacy and
  statistical minimax rates,'' in \emph{2013 IEEE 54th Annual Symposium on
  Foundations of Computer Science}, 2013, pp. 429--438.

\bibitem{AndresKostasCatuscia_GeoInd}
M.~E. Andr{\'e}s, N.~E. Bordenabe, K.~Chatzikokolakis, and C.~Palamidessi,
  ``Geo-indistinguishability: Differential privacy for location-based
  systems,'' in \emph{Proceedings of the 2013 ACM SIGSAC conference on Computer
  \& communications security}, 2013, pp. 901--914.

\bibitem{Bordenabe:14:CCS}
N.~E. Bordenabe, K.~Chatzikokolakis, and C.~Palamidessi, ``Optimal
  geo-indistinguishable mechanisms for location privacy,'' in \emph{Proceedings
  of the 21th ACM Conference on Computer and Communications Security (CCS
  2014)}, 2014.

\bibitem{Fernandes:21:LICS}
N.~Fernandes, A.~McIver, and C.~Morgan, ``The laplace mechanism has optimal
  utility for differential privacy over continuous queries,'' in \emph{36th
  Annual {ACM/IEEE} Symposium on Logic in Computer Science, {LICS} 2021.}\hskip
  1em plus 0.5em minus 0.4em\relax {IEEE}, 2021, pp. 1--12.

\bibitem{gillam2018exploring}
L.~Gillam, K.~Katsaros, M.~Dianati, and A.~Mouzakitis, ``Exploring edges for
  connected and autonomous driving,'' in \emph{IEEE INFOCOM 2018-IEEE
  Conference on Computer Communications Workshops (INFOCOM WKSHPS)}.\hskip 1em
  plus 0.5em minus 0.4em\relax IEEE, 2018, pp. 148--153.

\bibitem{maple2019connected}
C.~Maple, M.~Bradbury, A.~T. Le, and K.~Ghirardello, ``A connected and
  autonomous vehicle reference architecture for attack surface analysis,''
  \emph{Applied Sciences}, vol.~9, no.~23, p. 5101, 2019.

\bibitem{hahn2019security}
D.~Hahn, A.~Munir, and V.~Behzadan, ``Security and privacy issues in
  intelligent transportation systems: Classification and challenges,''
  \emph{IEEE Intelligent Transportation Systems Magazine}, vol.~13, no.~1, pp.
  181--196, 2019.

\bibitem{lin2013achieving}
X.~Lin and X.~Li, ``Achieving efficient cooperative message authentication in
  vehicular ad hoc networks,'' \emph{IEEE Transactions on Vehicular
  Technology}, vol.~62, no.~7, pp. 3339--3348, 2013.

\bibitem{kumar2015intelligent}
N.~Kumar, R.~Iqbal, S.~Misra, and J.~J. Rodrigues, ``An intelligent approach
  for building a secure decentralized public key infrastructure in vanet,''
  \emph{Journal of Computer and System Sciences}, vol.~81, no.~6, pp.
  1042--1058, 2015.

\bibitem{zhao2020survey}
P.~Zhao, G.~Zhang, S.~Wan, G.~Liu, and T.~Umer, ``A survey of local
  differential privacy for securing internet of vehicles,'' \emph{The Journal
  of Supercomputing}, vol.~76, no.~11, pp. 8391--8412, 2020.

\bibitem{franke2013understanding}
T.~Franke and J.~F. Krems, ``Understanding charging behaviour of electric
  vehicle users,'' \emph{Transportation Research Part F: Traffic Psychology and
  Behaviour}, vol.~21, pp. 75--89, 2013.

\bibitem{kumar2020adoption}
R.~R. Kumar and K.~Alok, ``Adoption of electric vehicle: A literature review
  and prospects for sustainability,'' \emph{Journal of Cleaner Production},
  vol. 253, p. 119911, 2020.

\bibitem{tian2016real}
Z.~Tian, T.~Jung, Y.~Wang, F.~Zhang, L.~Tu, C.~Xu, C.~Tian, and X.-Y. Li,
  ``Real-time charging station recommendation system for electric-vehicle
  taxis,'' \emph{IEEE Transactions on Intelligent Transportation Systems},
  vol.~17, no.~11, pp. 3098--3109, 2016.

\bibitem{zhang2021intelligent}
W.~Zhang, H.~Liu, F.~Wang, T.~Xu, H.~Xin, D.~Dou, and H.~Xiong, ``Intelligent
  electric vehicle charging recommendation based on multi-agent reinforcement
  learning,'' in \emph{Proceedings of the Web Conference 2021}, 2021, pp.
  1856--1867.

\bibitem{flocea2022electric}
R.~Flocea, A.~H{\^\i}ncu, A.~Robu, S.~Senocico, A.~Traciu, B.~M. Remus, M.~S.
  R{\u{a}}boac{\u{a}}, and C.~Filote, ``Electric vehicle smart charging
  reservation algorithm,'' \emph{Sensors}, vol.~22, no.~8, p. 2834, 2022.

\bibitem{wang2019sharedcharging}
G.~Wang, W.~Li, J.~Zhang, Y.~Ge, Z.~Fu, F.~Zhang, Y.~Wang, and D.~Zhang,
  ``sharedcharging: Data-driven shared charging for large-scale heterogeneous
  electric vehicle fleets,'' \emph{Proceedings of the ACM on Interactive,
  Mobile, Wearable and Ubiquitous Technologies}, vol.~3, no.~3, pp. 1--25,
  2019.

\bibitem{plugshare}
\BIBentryALTinterwordspacing
``Privacy policy.'' [Online]. Available:
  \url{https://company.plugshare.com/privacy.html}
\BIBentrySTDinterwordspacing

\bibitem{chargepoint}
\BIBentryALTinterwordspacing
``Chargepoint privacy and cookie policy for europe.'' [Online]. Available:
  \url{https://eu.chargepoint.com/privacy_policy}
\BIBentrySTDinterwordspacing

\bibitem{chow2009casper}
C.-Y. Chow, M.~F. Mokbel, and W.~G. Aref, ``Casper* query processing for
  location services without compromising privacy,'' \emph{ACM Transactions on
  Database Systems (TODS)}, vol.~34, no.~4, pp. 1--48, 2009.

\bibitem{niu2014achieving}
B.~Niu, Q.~Li, X.~Zhu, G.~Cao, and H.~Li, ``Achieving k-anonymity in
  privacy-aware location-based services,'' in \emph{IEEE INFOCOM 2014-IEEE
  Conference on Computer Communications}.\hskip 1em plus 0.5em minus
  0.4em\relax IEEE, 2014, pp. 754--762.

\bibitem{DworkDP_Compositionality}
C.~Dwork and A.~Roth, ``The algorithmic foundations of differential privacy,''
  \emph{Found. Trends Theor. Comput. Sci.}, vol.~9, no. 3–4, p. 211–407,
  aug 2014.

\bibitem{guo2018independent}
N.~Guo, L.~Ma, and T.~Gao, ``Independent mix zone for location privacy in
  vehicular networks,'' \emph{IEEE Access}, vol.~6, pp. 16\,842--16\,850, 2018.

\bibitem{CacheLocal}
S.~Amini, J.~Lindqvist, J.~I. Hong, M.~Mou, R.~Raheja, J.~Lin, N.~Sadeh, and
  E.~Tochb, ``Cach\'{e}: Caching location-enhanced content to improve user
  privacy,'' \emph{SIGMOBILE Mob. Comput. Commun. Rev.}, vol.~14, no.~3, p.
  19–21, dec 2011.

\bibitem{ChallengesOfLocal}
P.~Asuquo, H.~Cruickshank, J.~Morley, C.~P.~A. Ogah, A.~Lei, W.~Hathal, S.~Bao,
  and Z.~Sun, ``Security and privacy in location-based services for vehicular
  and mobile communications: An overview, challenges, and countermeasures,''
  \emph{IEEE Internet of Things Journal}, vol.~5, no.~6, pp. 4778--4802, 2018.

\bibitem{zhou2018achieving}
L.~Zhou, L.~Yu, S.~Du, H.~Zhu, and C.~Chen, ``Achieving differentially private
  location privacy in edge-assistant connected vehicles,'' \emph{IEEE Internet
  of Things Journal}, vol.~6, no.~3, pp. 4472--4481, 2018.

\bibitem{luo2019geo}
L.~Luo, Z.~Han, C.~Xu, and G.~Zhao, ``A geo-indistinguishable location privacy
  preservation scheme for location-based services in vehicular networks,'' in
  \emph{International Conference on Algorithms and Architectures for Parallel
  Processing}.\hskip 1em plus 0.5em minus 0.4em\relax Springer, 2019, pp.
  610--623.

\bibitem{qiu2020location}
C.~Qiu, A.~C. Squicciarini, C.~Pang, N.~Wang, and B.~Wu, ``Location privacy
  protection in vehicle-based spatial crowdsourcing via
  geo-indistinguishability,'' \emph{IEEE Transactions on Mobile Computing}, pp.
  1--1, 2020.

\bibitem{Cunningham2021Trace}
T.~Cunningham, G.~Cormode, H.~Ferhatosmanoglu, and D.~Srivastava, ``Real-world
  trajectory sharing with local differential privacy,'' \emph{Proc. VLDB
  Endow.}, vol.~14, no.~11, p. 2283–2295, jul 2021.

\bibitem{AgarwalIBU}
D.~Agrawal and C.~C. Aggarwal, ``On the design and quantification of privacy
  preserving data mining algorithms,'' in \emph{Proceedings of the Twentieth
  ACM SIGMOD-SIGACT-SIGART Symposium on Principles of Database Systems}, ser.
  PODS '01.\hskip 1em plus 0.5em minus 0.4em\relax New York, NY, USA:
  Association for Computing Machinery, 2001, p. 247–255.

\bibitem{agrawal2005privacy}
R.~Agrawal, R.~Srikant, and D.~Thomas, ``Privacy preserving olap,'' in
  \emph{Proceedings of the 2005 ACM SIGMOD international conference on
  Management of data}, 2005, pp. 251--262.

\bibitem{EhabConvergenceIBU}
\BIBentryALTinterwordspacing
E.~ElSalamouny and C.~Palamidessi, ``Full convergence of the iterative bayesian
  update and applications to mechanisms for privacy protection,'' \emph{CoRR},
  vol. abs/1909.02961, 2019. [Online]. Available:
  \url{http://arxiv.org/abs/1909.02961}
\BIBentrySTDinterwordspacing

\bibitem{EhabGIBU}
------, ``Generalized iterative bayesian update and applications to mechanisms
  for privacy protection,'' in \emph{2020 IEEE European Symposium on Security
  and Privacy (EuroS\&P)}.\hskip 1em plus 0.5em minus 0.4em\relax IEEE, 2020,
  pp. 490--507.

\bibitem{ChatzikokolakisElSalamounyPalamidessi+2017+308+328}
K.~Chatzikokolakis, E.~ElSalamouny, and C.~Palamidessi, ``Efficient utility
  improvement for location privacy,'' \emph{Proceedings on Privacy Enhancing
  Technologies}, vol. 2017, no.~4, pp. 308--328, 2017.

\bibitem{ghosh2009universally}
A.~Ghosh, T.~Roughgarden, and M.~Sundararajan, ``Universally utility-maximizing
  privacy mechanisms,'' 2009.

\bibitem{bulut2017mitigating}
E.~Bulut and M.~C. Kisacikoglu, ``Mitigating range anxiety via
  vehicle-to-vehicle social charging system,'' in \emph{2017 IEEE 85th
  Vehicular Technology Conference (VTC Spring)}.\hskip 1em plus 0.5em minus
  0.4em\relax IEEE, 2017, pp. 1--5.

\bibitem{Duan2020emerging}
W.~Duan, J.~Gu, M.~Wen, G.~Zhang, Y.~Ji, and S.~Mumtaz, ``Emerging technologies
  for 5g-iov networks: Applications, trends and opportunities,'' \emph{IEEE
  Network}, vol.~34, no.~5, pp. 283--289, 2020.

\bibitem{ji2020Artificial}
H.~Ji, O.~Alfarraj, and A.~Tolba, ``Artificial intelligence-empowered edge of
  vehicles: Architecture, enabling technologies, and applications,'' \emph{IEEE
  Access}, vol.~8, pp. 61\,020--61\,034, 2020.

\bibitem{patil2020grid}
H.~Patil and V.~N. Kalkhambkar, ``Grid integration of electric vehicles for
  economic benefits: A review,'' \emph{Journal of Modern Power Systems and
  Clean Energy}, vol.~9, no.~1, pp. 13--26, 2020.

\bibitem{chargemap}
\BIBentryALTinterwordspacing
``Charging station map for electric cars.'' [Online]. Available:
  \url{https://chargemap.com/map}
\BIBentrySTDinterwordspacing

\bibitem{Lee2018exploring}
L.~Gillam, K.~Katsaros, M.~Dianati, and A.~Mouzakitis, ``Exploring edges for
  connected and autonomous driving,'' in \emph{IEEE INFOCOM 2018 - IEEE
  Conference on Computer Communications Workshops (INFOCOM WKSHPS)}, 2018, pp.
  148--153.

\bibitem{seo2016lte}
H.~Seo, K.-D. Lee, S.~Yasukawa, Y.~Peng, and P.~Sartori, ``Lte evolution for
  vehicle-to-everything services,'' \emph{IEEE communications magazine},
  vol.~54, no.~6, pp. 22--28, 2016.

\bibitem{paverd2014modelling}
A.~Paverd, A.~Martin, and I.~Brown, ``Modelling and automatically analysing
  privacy properties for honest-but-curious adversaries,'' \emph{Tech. Rep},
  2014.

\bibitem{OpenStreetMap}
{OpenStreetMap contributors}, ``{Planet dump retrieved from
  https://planet.osm.org },'' \url{ https://www.openstreetmap.org }, 2022.

\bibitem{chargingstations}
\BIBentryALTinterwordspacing
``Alternative fuels data center: Data downloads.'' [Online]. Available:
  \url{https://afdc.energy.gov/data_download/}
\BIBentrySTDinterwordspacing

\bibitem{chargingstations2}
\BIBentryALTinterwordspacing
``Datasf.'' [Online]. Available:
  \url{https://data.sfgov.org/browse?Department-Metrics_Publishing-Department=Municipal+Transportation+Agency+}
\BIBentrySTDinterwordspacing

\bibitem{piorkowski2009crawdad}
M.~Piorkowski, N.~Sarafijanovic-Djukic, and M.~Grossglauser, ``Crawdad data set
  epfl/mobility (v. 2009-02-24),'' 2009.

\bibitem{Li_PerturbationHidden}
X.~Li, Y.~Ren, L.~T. Yang, N.~Zhang, B.~Luo, J.~Weng, and X.~Liu,
  ``Perturbation-hidden: Enhancement of vehicular privacy for location-based
  services in internet of vehicles,'' \emph{IEEE Transactions on Network
  Science and Engineering}, vol.~8, no.~3, pp. 2073--2086, 2021.

\bibitem{zhang2018privacy}
L.~Zhang, X.~Meng, K.-K.~R. Choo, Y.~Zhang, and F.~Dai, ``Privacy-preserving
  cloud establishment and data dissemination scheme for vehicular cloud,''
  \emph{IEEE Transactions on Dependable and Secure Computing}, vol.~17, no.~3,
  pp. 634--647, 2018.

\bibitem{Freudiger:109437}
J.~Freudiger, M.~Raya, M.~F{\'e}legyh{\'a}zi, P.~Papadimitratos, and J.-P.
  Hubaux, ``Mix-zones for location privacy in vehicular networks,'' in
  \emph{ACM Workshop on Wireless Networking for Intelligent Transportation
  Systems (WiN-ITS)}, no. CONF, 2007.

\bibitem{Jiang_LPPM}
H.~Jiang, J.~Li, P.~Zhao, F.~Zeng, Z.~Xiao, and A.~Iyengar, ``Location
  privacy-preserving mechanisms in location-based services: A comprehensive
  survey,'' \emph{ACM Comput. Surv.}, vol.~54, no.~1, jan 2021.

\bibitem{zang2011anonymization}
H.~Zang and J.~Bolot, ``Anonymization of location data does not work: A
  large-scale measurement study,'' in \emph{Proceedings of the 17th annual
  international conference on Mobile computing and networking}, 2011, pp.
  145--156.

\bibitem{hecht2021predicting}
C.~Hecht, J.~Figgener, and D.~U. Sauer, ``Predicting electric vehicle charging
  station availability using ensemble machine learning,'' \emph{Energies},
  vol.~14, no.~23, p. 7834, 2021.

\bibitem{nait2018prediction}
A.~Nait-Sidi-Moh, A.~Ruzmetov, M.~Bakhouya, Y.~Naitmalek, and J.~Gaber, ``A
  prediction model of electric vehicle charging requests,'' \emph{Procedia
  Computer Science}, vol. 141, pp. 127--134, 2018.

\bibitem{ma2022multistep}
T.-Y. Ma and S.~Faye, ``Multistep electric vehicle charging station occupancy
  prediction using hybrid lstm neural networks,'' \emph{Energy}, vol. 244, p.
  123217, 2022.

\bibitem{almaghrebi2020data}
A.~Almaghrebi, F.~Aljuheshi, M.~Rafaie, K.~James, and M.~Alahmad, ``Data-driven
  charging demand prediction at public charging stations using supervised
  machine learning regression methods,'' \emph{Energies}, vol.~13, no.~16, p.
  4231, 2020.

\bibitem{luo2021deep}
R.~Luo, Y.~Zhang, Y.~Zhou, H.~Chen, L.~Yang, J.~Yang, and R.~Su, ``Deep
  learning approach for long-term prediction of electric vehicle (ev) charging
  station availability,'' in \emph{2021 IEEE International Intelligent
  Transportation Systems Conference (ITSC)}.\hskip 1em plus 0.5em minus
  0.4em\relax IEEE, 2021, pp. 3334--3339.

\end{thebibliography}

\end{document}